\documentclass[10.5pt]{article}
\usepackage{scicite}

\newcounter{lastnote}

\usepackage{amsmath}
\usepackage{amssymb}
\usepackage{algorithm,algpseudocode}
\usepackage{algpseudocode}
\usepackage{nicefrac}
\usepackage{pdflscape}
\usepackage{array}
\usepackage{placeins}
\usepackage{hyperref}
\DeclareMathOperator*{\arginf}{arg\,inf}

\usepackage{latexsym}
\usepackage{stmaryrd}
\usepackage{amsthm}
\usepackage{float}
\usepackage{caption}
\usepackage[autostyle=true]{csquotes}
\usepackage{mathtools}
\usepackage{tikz}

\newcommand{\ci}{\perp\!\!\!\perp}
\usepackage{mathabx}
\usetikzlibrary{patterns}
\usetikzlibrary{shapes.arrows}
\newtheorem{Conjecture}{Conjecture}
\newtheorem{Theorem}{Theorem}[section]
\newtheorem{Lemma}{Lemma}

\newtheorem{Proposition}{Proposition}

\theoremstyle{definition}
\newtheorem{Definition}{Definition}
\newtheorem{Property}{Property}
\newtheorem{Example}{Example}
\usepackage[backend=biber,natbib=true,maxnames=25]{biblatex}
\title{Complexity as Causal Information Integration}

\author{Carlotta Langer$^{1 }$ and Nihat Ay $^{1, 2, 3}$
\\
\\
\normalsize{$^{1}$ \quad Max Planck Institute for Mathematics in the Sciences, Leipzig, Germany }\\
\normalsize{$^{2}$ \quad Leipzig University, Leipzig, Germany} \\
\normalsize{$^{3}$ \quad Santa Fe Institute, Santa Fe, USA}}

\date{}

\begin{document}
\maketitle 
\begin{abstract}
Complexity measures in the context of the Integrated Information Theory of consciousness try to quantify the strength of the causal connections between different neurons. This is done by minimizing the KL-divergence between a full system and one without causal cross-connections. Various measures have been proposed and compared in this setting. We will discuss a class of information geometric measures that aim at assessing the intrinsic causal cross-influences in a system. One promising candidate of these measures, denoted by $\Phi_{CIS}$, is based on conditional independence statements and does satisfy all of the properties that have been postulated as desirable. Unfortunately it does not have a graphical representation, which makes it less intuitive and difficult to analyze.
We propose an alternative approach using a latent variable, which models a common exterior influence. This leads to a measure $\Phi_{CII}$, Causal Information Integration, that satisfies all of the required conditions. Our measure can be calculated using an iterative information geometric algorithm, the em-algorithm. Therefore we are able to compare its behavior to existing integrated information measures.
\end{abstract}

\textit{keywords:} Complexity; Integrated Information; Causality; Conditional Independence; em-Algorithm
\section{Introduction}
The theory of Integrated Information aims at quantifying the amount and quality of consciousness of a neural network. 
It was originally proposed by Tononi and went through various phases of evolution, starting with one of the first papers "Consciousness and {Complexity"} \cite{Tononi1999} in 1999 to "Consciousness as Integrated Information--a Provisional {Manifesto"} \cite{Tononi} in 2008 and Integrated Information Theory (IIT) 3.0 \cite{IIT3.0} in 2014 to ongoing research. Although important parts of the methodology of this theory changed or got extended the two key concepts determining consciousness that virtually stayed fixed are "Information" and "Integration". 
Information refers to the number of different states a system can be in and Integration describes the amount to which the information is integrated among different parts of it. 
{Tononi summarizes this idea in Reference \cite{Tononi} with the following sentence:
\begin{quote}
\textit{In short, integrated information captures the information
generated by causal interactions in the whole, over and
above the information generated by the parts.}
\end{quote}    
Therefore Integrated Information can be seen as a measure of the systems complexity. In this context it belongs to the class of theories that define complexity as to what extent the whole is more than the sum of its parts. } 

There are various ways to define a split system and the difference between them. Therefore, there exist different branches of complexity measures in the context of Integrated Information. {The most recent theory, IIT 3.0 \cite{IIT3.0}, goes far beyond the original measures and includes a different level of definitions corresponding to the quality of the measured consciousness, including the maximally irreducible conceptual structure (MICS) and the integrated conceptual information. In order to focus on the information geometric aspects of IIT, we follow the strategy of Oizumi et al.~\cite{uniframe} and Amari et al.~\cite{GeomInfInt}, restricting attention to measuring the integrated information in discrete $n$-dimensional stationary Markov processes from an information geometric point of view.}

In detail we will measure the distance between the full and the split system using the KL-divergence as proposed in Reference \cite{stochIntPreprint}, {published in Reference \cite{stochInt}}. This framework was further discussed in Reference \cite{GeomAppr}. Oizumi et al.~\cite{uniframe} and Amari et al.~\cite{GeomInfInt} summarize these ideas and add a Markov condition and an upper bound to clarify what a complexity measure should satisfy. {The Markov condition intends to model the removal of certain cross-time connections, which we call causal cross-connections. These connections are the ones that integrate information among the different nodes across different points in time. The upper bound was originally proposed in Reference \cite{decoding} and is given by the mutual information, which aims at quantifying the total information flow from one timestep to the next. These conditions are defined as necessary and do not specify a measure uniquely.} We will discuss the conditions in the next section. 

Additionally  Oizumi et al.~\cite{uniframe} and Amari et al.~\cite{GeomInfInt} introduce one measure that satisfies all of these requirements. This measure is described by conditional independence statements and will be denoted here by $\Phi_{CIS}$. We will introduce $\Phi_{CIS}$ along with two other existing measures, namely Stochastic Interaction $\Phi_{SI}$ \cite{stochInt} and Geometric Integrated Information $\Phi_{G}$ \cite{amari}. {The measure $\Phi_{SI}$ is not bounded from above by the mutual information and $\Phi_{G}$ does not satisfy the postulated Markov condition. }

Although $\Phi_{CIS}$ fits perfectly in the proposed framework, this measure does not correspond to a graphical representation and it is therefore difficult to analyze the causal nature of the measured information flow. {We focus on the notion of causality defined by Pearl in Reference \cite{pearl}, in which the correspondence between conditional independence statements and graphs, for instance DAGs or more generally chain graphs, is a key concept.
Moreover, we demonstrate that it is not possible to express the conditional independence statements corresponding to $\Phi_{CIS}$ using a chain graph even after adding latent variables. Following the reasoning of Pearls causality theory, however, this would be a desirable property.}

The main purpose of this paper is to propose a more intuitive approach {that ensures the consistency between graphical representation and conditional independence statements. This is achieved by using a latent variable that models a common exterior influence. Doing so leads to a new measure, which we call Causal Information Integration $\Phi_{CII}$.} This measure is specifically created to only measure the intrinsic causal cross-influences { in a setting with an unknown exterior influence} and it satisfies all the required conditions postulated by Oizumi et al. To assume the existence of an unknown exterior influence is not unreasonable, in fact one point of criticism concerning $\Phi_{SI}$ is that this measure does not account for exterior influences and therefore measures them erroneously as internal, see Section 6.9.~in Reference \cite{amari}. In a setting with known external influences, these can be integrated in the model as visible variables. This leads to a model discussed in Section \ref{Sectgroundtruth} that we call $\Phi_{T}$, which is an upper bound for $\Phi_{CII}$.

We discuss the relationships between the introduced measures in Section \ref{SectRel} and present a way of calculating $\Phi_{CII}$ by using an iterative information geometric algorithm, the em-algorithm described in Section \ref{SectEm}. { This algorithm is guaranteed to converge to a minimum, but this might be a local minimum. Therefore we have to run the algorithm multiple times to find a global minimum.}
Utilizing this algorithm we are able to compare the behavior of $\Phi_{CII}$ to existing integrated information measures.
\newpage
\subsection{Integrated Information Measures} \label{sectIntInfMeas}

Measures corresponding to Integrated Information investigate the information flow in a system from a time $t$ to $t+1$.
This flow is represented by the connections from the nodes $X_{i}$ in $t$ to the nodes $Y_{i}$ in $t+1, \, i \in \{1, \dots, n\}$  as displayed in Figure \ref{interactInt}. 
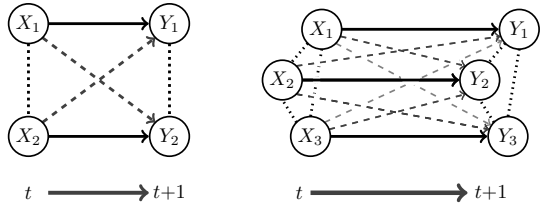
\begin{figure}[H]
\centering
\scalebox{0.75}{
\begin{tikzpicture}[rounded corners]
\draw[] (0,0) node {$X_{1}$};
\draw[->,line width=0.5mm] (0.35,0)--(2.133,0);
\draw[->,line width=0.5mm] (0.35,-2)--(2.133,-2);
\draw[->,line width=0.7mm, darkgray] (0.35,-3)--(2.133,-3);
\draw[->,dashed, darkgray, line width=0.5mm] (0.25,-0.25)--(2.25,-1.75);
\draw[->,dashed, darkgray,line width=0.5mm] (0.25,-1.75)--(2.25,-0.25);
\draw[line width=0.5mm,dotted] (0,-0.35)--(0,-1.65);
\draw[line width=0.5mm,dotted] (2.5,-0.35)--(2.5,-1.65);
\draw[] (0,-2) node {$X_{2}$};
\draw[] (2.5,0) node { $Y_{1}$};
\draw[] (2.5,-2) node { $Y_{2}$};
\draw[] (0, -3) node {$t$};
\draw[] (2.5,-3) node {$t$+$1$};
\draw[line width=0.3mm] (0,0) circle (10pt);
\draw[line width=0.3mm] (0,-2) circle (10pt);
\draw[line width=0.3mm] (2.5,0) circle (10pt);
\draw[line width=0.3mm] (2.5,-2) circle (10pt);
\draw[]  (5.2,-0.1) node (A) {$X_{1}$};
\draw[] (5,-2) node (B) {$X_{3}$};
\draw[] (4.5,-1) node (C) { $X_{2}$};
\draw[] (8.7,-0.1) node (D) {$Y_{1}$};
\draw[] (8.5,-2) node (E){$Y_{3}$};
\draw[] (8,-1) node (F) {$Y_{2}$};
\draw[->,thick, line width=0.5mm] (5.55,-0.1)--(8.333,-0.1);
\draw[->,thick,line width=0.5mm] (5.35,-2)--(8.133,-2);
\draw[->,dashed,gray,line width=0.3mm] (5.45,-0.35) coordinate (a_1) -- (8.2,-1.8) coordinate (a_2);
\draw[->, dashed, gray,line width=0.3mm] (5.25,-1.75) coordinate (b_1) --(8.4,-0.3) coordinate (b_2);
\draw[->,dashed, gray,line width=0.3mm] (4.75,-0.75) coordinate (c_1) --(8.366,-0.225) coordinate (c_2);
\draw[->,dashed,line width=0.3mm] (4.75,-1.25) coordinate (f_1)--(8.166,-1.875) coordinate (f_2);
\draw[->,dashed,gray,line width=0.3mm] coordinate (d_1) (5.53,-0.225)--(7.75,-0.75) coordinate (d_2);
\draw[->,dashed,gray,line width=0.3mm] (5.33,-1.875) coordinate (e_1)--(7.75,-1.25) coordinate (e_2);
\draw[thick, dotted, line width=0.5mm] (5.2,-0.45) coordinate (g_1) --(5,-1.65) coordinate (g_2);
\draw[thick,->,line width=0.5mm] (4.85,-1) coordinate (h_1)--(7.633,-1) coordinate (h_2);
\draw[thick, dotted, line width=0.5mm] (4.95,-0.35)--(4.65,-0.7);
\draw[thick, dotted,line width=0.5mm] (4.75,-1.75)--(4.5,-1.35);
\draw[thick, dotted,line width=0.5mm] (8.7,-0.45) --(8.5,-1.65);
\draw[thick, dotted,line width=0.5mm] (8.45,-0.35)--(8.15,-0.7);
\draw[thick, dotted,line width=0.5mm] (8.25,-1.75)--(8,-1.35);
\coordinate (c) at (intersection of a_1--a_2 and b_1--b_2);
\path (c) +(0,0.12) coordinate (m);
\fill[white] (m) circle (5pt);
\coordinate (d) at (intersection of a_1--a_2 and c_1--c_2);
\fill[white] (d) circle (3pt);
\coordinate (e) at (intersection of b_1--b_2 and d_1--d_2);
\path (e) +(0.06,0.06) coordinate (f);
\fill[white] (f) circle (3pt);
\coordinate (g) at (intersection of e_1--e_2 and a_1--a_2);
\fill[white] (g) circle (3pt);
\coordinate (h) at (intersection of b_1--b_2 and f_1--f_2);
\fill[white] (h) circle (3pt);
\coordinate (i) at (intersection of c_1--c_2 and g_1--g_2);
\fill[white] (i) circle (2pt);
\coordinate (j) at (intersection of f_1--f_2 and g_1--g_2);
\fill[white] (j) circle (2pt);
\coordinate (k) at (intersection of h_1--h_2 and g_1--g_2);
\fill[white] (k) circle (2pt);
\draw[->,dashed,darkgray,line width=0.3mm] (5.53,-0.225)--(7.75,-0.75);
\draw[->, dashed ,darkgray,line width=0.3mm] (5.33,-1.875)--(7.75,-1.25);
\draw[->,dashed,darkgray,line width=0.3mm] (4.75,-1.25)--(8.166,-1.875);
\draw[thick,->, line width=0.5mm] (4.85,-1)--(7.633,-1);

\draw[thick,->,darkgray, line width=0.8mm] (5,-3)--(7.783,-3);
\draw[] (4.8, -3) node{$t$};
\draw[] (8.2,-3) node{$t$+$1$};
\draw[->,darkgray, dashed,line width=0.3mm] (4.75,-0.75) --(8.366,-0.225);
\draw[line width=0.3mm] (4.5,-1) circle (10pt);
\draw[line width=0.3mm] (5.2,-0.1) circle (10pt);
\draw[line width=0.3mm] (5,-2) circle (10pt);
\draw[line width=0.3mm] (8.7,-0.1) circle (10pt);
\draw[line width=0.3mm] (8.5,-2) circle (10pt);
\draw[line width=0.3mm] (8,-1) circle (10pt);
\end{tikzpicture} }
\caption{The fully connected system for $n = 2$ and $n=3$. } \label{interactInt}
\end{figure}  
The systems are modeled as discrete, stationary, $n$-dimensional Markov processes $(Z_{t})_{t \in \mathbb{N}}$
\begin{equation*}
X = (X_{1}, \dots , X_{n}) = (X_{1,t}, \dots , X_{n,t}), \quad \quad Y = (Y_{1}, \dots , Y_{n}) = (X_{1,t+1}, \dots , X_{n,t+1}), \quad \quad Z = (X,Y)
\end{equation*}
on a finite set $\mathcal{Z} \neq \emptyset$, which is the Cartesian product of the sample  spaces of $X_{i}$ $i \in \{1 \dots n\}$ , denoted by $\mathcal{X}_{i}$
\begin{equation*}
\mathcal{Z} = \mathcal{X} \times \mathcal{Y} = \bigtimes\limits_{i=1}^{n} \mathcal{X}_{i} \times \bigtimes\limits_{i=1}^{n} \mathcal{Y}_{i}.
\end{equation*} 
{It is possible to apply the following methods to non-stationary distributions, but this assumption in addition to the process being Markovian allows us to restrict the discussion to one time step.

Let $MP(\mathcal{Z})$ be set of distributions that belong to these Markov processes. }

Denote the complement of $X_{i}$ in $ X $ by  
$X_{I \setminus \{i\}} = (X_{1}, \dots , X_{i-1},X_{i+1}, \dots , X_{n}) $ with $I = \{1, \dots, n\}$. Corresponding to this notation $x_{I \setminus \{i\}} \in \mathcal{X}_{I \setminus \{i\}}$ describes the elementary events of $X_{I \setminus \{i\}}$. We will use the analogue notation in the case of $Y$ and we will write $z \in \mathcal{Z}$ instead of  $ (x,y) \in \mathcal{X} \times \mathcal{Y} $. The set of probability distributions on $\mathcal{Z}$ will be denoted by $\mathcal{P}(\mathcal{Z})$. Throughout this article we will restrict attention to strictly positive distributions.

The core idea of measuring Integrated Information is to determine how much the initial system differs from one in which no information integration takes place. The former will be called a "full" system, because we allow all possible connections between the nodes, and the latter will be called a "split" system.
Graphical representations of the full systems for $n=2,3$ and their connections are depicted in Figure  \ref{interactInt}. 
In this article we are using graphs that describe the conditional independence structure of the corresponding sets of distributions. An introduction to those is given in Appendix \ref{AppGraph}. 

{Graphs are not only a tool to conveniently represent conditional independence statements, but the connection between conditional independence and graphs is a core concept of Pearls causality theory. The interplay between graphs and conditional independence statements provides a consistent foundation of causality. In Reference \cite{pearl} Section 1.3 Pearl emphasizes the importance  of a graphical representation with the following statement:
\begin{quote}
\textit{It seems that if conditional independence judgments are by-products of
stored causal relationships, then tapping and representing those relationships directly
would be a more natural and more reliable way of expressing what we know or believe
about the world. This is indeed the philosophy behind causal Bayesian networks.}
\end{quote} 
Therefore, measures of the strength of causal cross-connections should be based on split models, that have a graphical representation.} 

Following the concept introduced in References \cite{stochIntPreprint} \cite{stochInt}, the difference between the measures corresponding to the full and split systems will be calculated by using the KL-divergence.
\begin{Definition}[Complexity]
Let $\mathcal{M}$ be a set of probability distributions on $\mathcal{Z}$ corresponding to a split system. 
Then we minimize the KL-divergence between $\mathcal{M}$ and the distribution of the fully connected system $\tilde{P}$ to calculate the complexity
\begin{equation*}
\Phi_{\mathcal{M}} = \inf\limits_{Q \in \mathcal{M}} D_{\mathcal{Z}}(\tilde{P} \parallel Q)  = \sum\limits_{z \in \mathcal{Z}} \tilde{P}(z) \, log \dfrac{ \tilde{P}(z)}{Q(z)}.
\end{equation*}
\end{Definition}
Minimizing the KL-divergence with respect to the second argument is called $m$-projection or rI-projection. Hence we will call $P^{\star}$ with 
\begin{equation*}
P^{\star} = \arginf\limits_{Q \in \mathcal{M}} D_{\mathcal{Z}}(\tilde{P} \parallel Q)
\end{equation*} 
the projection of $\tilde{P}$ to $\mathcal{M}$.

The question remains how to define the split model $\mathcal{M}$. We want to measure the information that gets integrated between different nodes in different points in time. In Figure \ref{interactInt} these are the dashed connections, {also called cross-influences in Reference \cite{uniframe}. We will refer to the dashed connections as causal cross-connections}.

In order to ensure that these connections are removed in the split system, the authors of Reference \cite{uniframe} and Reference \cite{GeomInfInt} argue that $Y_{j}$ should be independent of $X_{i}$ given $X_{I \setminus \{i\}}$, $i \neq j$, leading to the following property.
\begin{Property} \label{postulate1}
A valid split system should satisfy the Markov condition
\begin{equation}
Q(X_{i},Y_{j} \mid X_{I \setminus \{i\}}) = Q(X_{i}\mid X_{I \setminus \{i\}})Q(Y_{j}\mid X_{I \setminus \{i\}}), \quad \, i \neq j \label{MarkovPop},
\end{equation}
with $Q \in \mathcal{P}(\mathcal{Z}) $. This can also be written in the following form
\begin{equation}
Y_{j} \ci X_{i} \vert X_{I \setminus \{i\}}.
\end{equation}
\end{Property} 
Now we take a closer look at the remaining connections.
The dotted lines connect nodes belonging to the same point in time. {These connections between the $Y_{i}$s might result from common internal influences, meaning a correlation between the $X_{i}$s passed on to the next point in time via the dashed or solid connections. Additionally Amari points out in Section 6.9 in Reference \cite{amari} that there might exist a common exterior influence on the $Y_{i}$s. Although the measured integrated information should be internal and independent of external influences, the system itself is in general not completely independent of its environment.  }

Since we want to measure the amount of integrated information between $t$ and $t+1$, the distribution in $t$, and therefore the connection between the $X_{i}$s, should stay unchanged in the split system. The dotted connections between the $Y_{i}$s play an important role in Property \ref{postulate2}. For this property, we will consider the split system in which the solid and dashed connections are removed.

The solid arrows represent the influence of a node in $t$ on itself in $t+1$ and removing these arrows, in addition to the causal cross-connections, leads to a system with completely disconnected points in time as shown on the right in {Figure} \ref{interiorexteriorInfluence}. The distributions corresponding to this split system are 
\begin{equation*}
\mathcal{M}_{I} =\{ Q \in \mathcal{P} (\mathcal{Z}) \vert Q(z) =Q(x) Q(y), \forall z = (x,y) \in \mathcal{Z}\}
\end{equation*} 
and the measure $\Phi_{I}$  is given by the mutual information $I(X;Y)$, which is defined in the following way
\begin{equation*}
\Phi_{I} = I(X;Y) = \sum\limits_{z \in \mathcal{Z}} P(x,y) \, log \left( \dfrac{P(x,y)}{P(x)P(y)} \right).
\end{equation*}
Since there is no information flow between the time steps Oizumi et al.~argue in Reference \cite{uniframe} that an integrated information measure should be bounded from above by the mutual information.

\begin{Property} \label{postulate2} The mutual information should be an upper bound for an Integrated Information measure
\begin{equation*}
\Phi_{\mathcal{M}} = \inf\limits_{Q \in \mathcal{M}} D_{\mathcal{Z}}( \tilde{P} \mid Q) \leq I(X;Y).
\end{equation*}
\end{Property}
Oizumi et al.~\cite{decoding},\cite{uniframe} and Amari et al.~\cite{GeomInfInt} state that this property is {natural, because an Integrated Information measure should be bounded by the total amount of information flow between the different points in time. The postulation of this property led to a discussion in Reference \cite{Comparison}. The point of disagreement concerns the edge between the $Y_{i}$s.  On the one hand this connection takes into account that the $Y_{i}$s might have a common exterior influence that affects all the $Y_{i}$s, as pointed out by Amari in Reference \cite{amari}. } This is symbolized by the additional node $W$ in Figure  \ref{interiorexteriorInfluence} and this should not contribute to the value of Integrated Information between the different points in time.
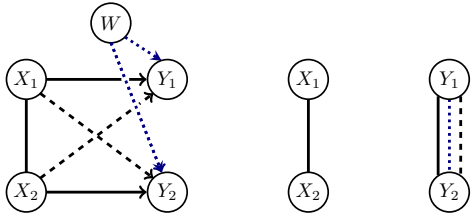
\begin{figure}[H]
\centering
\scalebox{0.75}{
\begin{tikzpicture}[rounded corners]
\draw[] (10,0) node {$X_{1}$};
\draw[line width=0.5mm] (10,-0.35)--(10,-1.65);
\draw[line width=0.5mm, blue!50!black, dotted] (12.5,-0.35)--(12.5,-1.65);
\draw[line width=0.5mm, dashed] (12.7,-0.3)--(12.7,-1.7);
\draw[line width=0.5mm] (12.3,-0.3)--(12.3,-1.7);
\draw[] (10,-2) node {$X_{2}$};
\draw[] (12.5,0) node { $Y_{1}$};
\draw[] (12.5,-2) node { $Y_{2}$};
\draw[line width=0.3mm] (10,0) circle (10pt);
\draw[line width=0.3mm] (10,-2) circle (10pt);
\draw[line width=0.3mm] (12.5,0) circle (10pt);
\draw[line width=0.3mm] (12.5,-2) circle (10pt);

\draw[line width=0.5mm] (5,0) node {$X_{1}$};
\draw[->,line width=0.5mm] (5.35,0) coordinate (b_1)--(7.133,0) coordinate (b_2);
\draw[->,line width=0.5mm] (5.35,-2)--(7.133,-2);
\draw[line width=0.3mm] (6.5,1) circle (10pt);
\draw[->,>=stealth,color =blue!50!black, dotted,line width=0.5mm] (6.75, 0.75) --(7.4,0.33);
\draw[->,color =blue!50!black,dotted,line width=0.5mm] (6.5, 0.65) coordinate (a_1) --(7.4,-1.65) coordinate (a_2);
\coordinate (k) at (intersection of b_1--b_2 and a_1--a_2);
\fill[white] (k) circle (2pt);
\draw[thick, line width=0.5mm] (5,-0.35) coordinate (g_1) --(5,-1.65) coordinate (g_2);
\draw[->, line width = 0.5mm, dashed] (5.25, -1.75)--(7.25, -0.25);
\draw[->, line width = 0.5mm, dashed] (5.25, -0.25)--(7.25, -1.75);
\fill[white] (6.95, -0.48) circle (2pt);
\draw[->,>=stealth,color = blue!50!black, dotted, line width=0.5mm] (6.5, 0.65) coordinate (a_1) --(7.4,-1.65) coordinate (a_2);
\draw[line width=0.5mm] (6.5,1) node { $W$};
\draw[line width=0.5mm] (5,-2) node {$X_{2}$};
\draw[line width=0.5mm] (7.5,0) node { $Y_{1}$};
\draw[line width=0.5mm] (7.5,-2) node { $Y_{2}$};
\draw[line width=0.3mm] (5,0) circle (10pt);
\draw[line width=0.3mm] (5,-2) circle (10pt);
\draw[line width=0.3mm] (7.5,0) circle (10pt);
\draw[line width=0.3mm] (7.5,-2) circle (10pt);
\end{tikzpicture} }
\caption{Interior and exterior influences on $Y$ in the full and the split system corresponding to $\Phi_{I}$.} \label{interiorexteriorInfluence}
\end{figure}
On the other hand, we know that if the $X_{i}$s are correlated, then the correlation is passed to the $Y_{i}$s via the solid and dashed arrows. {The edges created by calculating the marginal distribution on $Y$ also contain these correlations.} The question now is, how much of these correlations {integrate information in the system} and should therefore be measured. Kanwal et al.~discuss this problem in Reference \cite{Comparison}. They distinguish between intrinsic and extrinsic influences that cause the connections between the $Y_{i}$s in the way displayed in Figure \ref{interiorexteriorInfluence}. By calculating the split system for $\Phi_{I}$ the edge between the $Y_{i}$s might compensate for the solid arrows and common exterior influences, but also for the dashed, causal cross-connections, as shown in Figure \ref{interiorexteriorInfluence} on the right. Kanwal et al.~analyze an example of a full system without a common exterior influence with the result that there are cases in which a measure that only removes the causal cross-connections has a larger value than $\Phi_{I}$. This is only possible if the undirected edge between the $Y_{i}$s compensates a part of the causal cross-connections. Hence $\Phi_{I}$ does not measure all the intrinsic causal cross-influences. Therefore Kanwal et al.~question the use of the mutual information as an upper bound.

Then again, we would like to contribute a different perspective.
Admitting to Property \ref{postulate2} does not necessarily mean that the connections between the $Y_{i}$s are fixed. It may merely mean that $\mathcal{M}_{I}$ is a subset of the set of split distributions. 
We will see that the measures $\Phi_{CIS}$ and $\Phi_{CII}$ do satisfy Property \ref{postulate2} in this way. Although the argument that $\Phi_{I}$ measures all the intrinsic influences is no longer valid, satisfying Property \ref{postulate2} is still desirable in general. Consider an initial system with the distribution $\tilde{P}(z) = \tilde{P}(x) \tilde{P}(y), \, \forall z \in \mathcal{Z}$. This system has a common exterior influence on the $Y_{i}$s and no connection between the different points in time. Since there is no information flow between the points in time, a measure for Integrated Information $\Phi_{\mathcal{M}}$ should be zero for all distributions of this form. This is the case exactly when $\mathcal{M}_{I} \subseteq \mathcal{M}$, hence when $\Phi_{I}$ is an upper bound for $\Phi_{\mathcal{M}}$. In order to emphasize this point we propose a modified version of Property \ref{postulate2}.

\begin{Property}\label{postulate2neu}
The set $\mathcal{M}_{I}$ should be a subset of the split model $\mathcal{M}$ corresponding to the Integrated Information measure $\Phi_{\mathcal{M}}$. Then the inequality
\begin{equation*}
\Phi_{\mathcal{M}} = \inf\limits_{Q \in \mathcal{M}} D_{\mathcal{Z}}( \tilde{P} \mid Q) \leq I(X;Y)
\end{equation*}
holds.
\end{Property}

Note that the new formulation is stronger, hence Property \ref{postulate2} is a consequence of Property \ref{postulate2neu}. Every measure discussed here that satisfies Property \ref{postulate2} also fulfills Property \ref{postulate2neu}. Therefore we will keep referring to Property \ref{postulate2} in the following sections.

Figure \ref{overall} displays an overview over the different measures and whether they satisfy Properties \ref{postulate1} and \ref{postulate2}. 
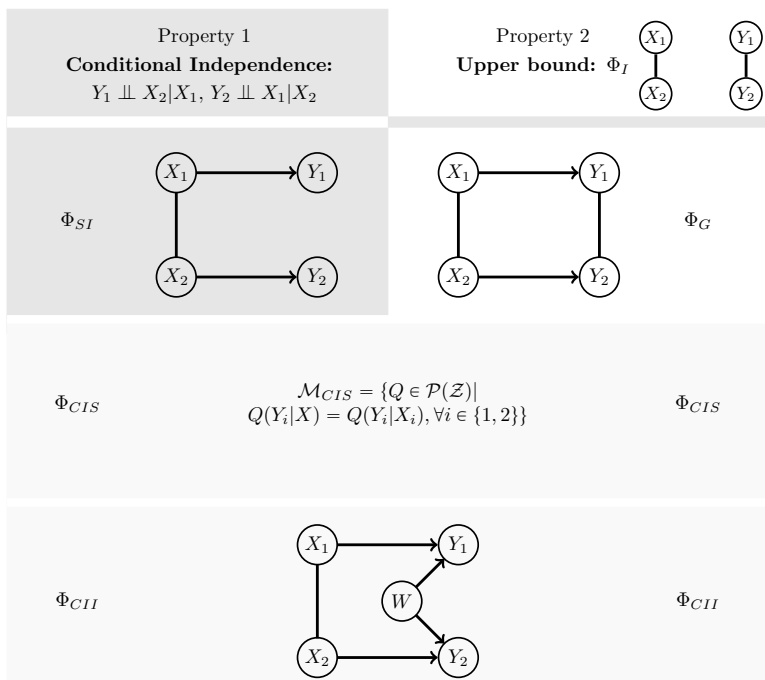
\begin{figure}[H]
\centering
\scalebox{.75}{
\begin{tikzpicture}
\draw[fill = gray!20, gray!20]  (-5.5,-4.25) rectangle (1.25,-10);
\draw[fill = gray!5, gray!5]  (-5.5,-9.75) rectangle (8,-16.25);
\draw[] (-2, -4.75) node {Property \ref{postulate1}};
\draw[] (4, -4.75) node {Property \ref{postulate2}};
\draw[] (4, -5.25) node {\textbf{Upper bound: }$\Phi_{I}$};
\draw[] ( -2, -5.25) node[ align=center] {\textbf{Conditional Independence: }};
\draw[] (6,-4.75) node {\small $X_{1}$};
\draw[] (6.75, -8) node{$\Phi_{G}$};
\draw[] (-4.25, -8) node{$\Phi_{SI}$};
\draw[] (-4.25, -11.25) node{$\Phi_{CIS}$};
\draw[] (-4.25, -14.75) node{$\Phi_{CII}$};
\draw[] (6.75, -11.25) node{$\Phi_{CIS}$};
\draw[] (6.75, -14.75) node{$\Phi_{CII}$};
\draw[line width=0.5mm] (6,-5.45)--(6,-5.05);
\draw[line width=0.5mm] (7.6, -5.45)--(7.6,-5.05);
\draw[] (6,-5.75) node {\small $X_{2}$};
\draw[] (7.6,-4.75) node {\small $Y_{1}$};
\draw[] (7.6,-5.75) node {\small $Y_{2}$};
\draw[line width=0.3mm] (6,-4.75) circle (8pt);
\draw[line width=0.3mm] (6,-5.75) circle (8pt);
\draw[line width=0.3mm] (7.6,-4.75) circle (8pt);
\draw[line width=0.3mm] (7.6,-5.75) circle (8pt);

\draw[] (-2, -5.75) node {$Y_{1} \ci X_{2} \vert X_{1},$ $Y_{2} \ci X_{1} \vert X_{2}$};
\draw[] (-2.5,-7.15) node {$X_{1}$};
\draw[line width=0.5mm] (-2.5,-7.5)--(-2.5,-8.65);
\draw[] (-2.5,-9) node {$X_{2}$};
\draw[] (0,-7.15) node { $Y_{1}$};
\draw[] (0,-9) node { $Y_{2}$};
\draw[line width=0.3mm] (-2.5,-7.15) circle (10pt);
\draw[line width=0.3mm] (-2.5,-9) circle (10pt);
\draw[line width=0.3mm] (0,-7.15) circle (10pt);
\draw[line width=0.3mm] (0,-9) circle (10pt);
\draw[->,line width=0.5mm] (-2.15,-7.15)--(-0.366,-7.15);
\draw[->,line width=0.5mm] (-2.15,-9)--(-0.366,-9);

\draw[] (2.5,-7.15) node {$X_{1}$};
\draw[line width=0.5mm] (2.5,-7.5)--(2.5,-8.65);
\draw[line width=0.5mm] (5,-7.5)--(5,-8.65);
\draw[] (2.5,-9) node {$X_{2}$};
\draw[] (5,-7.15) node { $Y_{1}$};
\draw[] (5,-9) node { $Y_{2}$};
\draw[line width=0.3mm] (2.5,-7.15) circle (10pt);
\draw[line width=0.3mm] (2.5,-9) circle (10pt);
\draw[line width=0.3mm] (5,-7.15) circle (10pt);
\draw[line width=0.3mm] (5,-9) circle (10pt);
\draw[->,line width=0.5mm] (2.85,-7.15)--(4.633,-7.15);
\draw[->,line width=0.5mm] (2.85,-9)--(4.633,-9);

\draw[] (1.25, -11.25) node[align = center] {$\mathcal{M}_{CIS} = \{ Q \in \mathcal{P}( \mathcal{Z}) \vert $ \\ $ Q(Y_{i} \vert X) = Q(Y_{i} \vert X_{i} ), \forall i \in \{1,2\} \}$};

\draw[->,line width=0.5mm] (0.35,-13.75)--(2.133,-13.75);
\draw[->,line width=0.5mm] (0.35,-15.75)--(2.133,-15.75);
\draw[->,line width=0.5mm] (1.75,-14.5)--(2.25,-14);
\draw[->,line width=0.5mm] (1.75,-15)--(2.25,-15.5);
\draw[ white, line width=2mm] (-5.5, -6.25) -- (6.25,-6.25);
\draw[ gray!20, line width=2mm] (1.25, -6.25) -- (8,-6.25);
\draw[ white, line width=1.5mm] (-5.5, -9.75) -- (8,-9.75);
\draw[ white, line width=1.5mm ] (-5.5, -13) -- (8,-13);
\draw[line width=0.5mm] (0,-14.1)--(0,-15.41);
\draw[] (0,-15.75) node {$X_{2}$};
\draw[] (0,-13.75) node {$X_{1}$};
\draw[] (2.5,-13.75) node { $Y_{1}$};
\draw[] (2.5,-15.75) node { $Y_{2}$};
\draw[] (1.5, -14.75) node{ $W$ };
\draw[line width=0.3mm] (0,-13.75) circle (10pt);
\draw[line width=0.3mm] (0,-15.75) circle (10pt);
\draw[line width=0.3mm] (2.5,-13.75) circle (10pt);
\draw[line width=0.3mm] (2.5,-15.75) circle (10pt);
\draw[line width=0.3mm] (1.5,-14.75) circle (10pt);
\end{tikzpicture}
}
\caption{The different measures and their properties in the case of $n = 2$.} \label{overall}
\end{figure} 
The first complexity measure that we are discussing does not fulfill Property \ref{postulate2}. It is  called Stochastic Interaction and was introduced by Ay in Reference \cite{stochIntPreprint} in 2001, later published in Reference \cite{stochInt}. Barrett and Seth discuss it in Reference \cite{Barrett} in the context of Integrated Information. In Reference \cite{GeomInfInt} the corresponding model is called  "fully split model". 

The core idea is to allow only the connections among the random variables in $t$ and additionally the connections between $X_{i}$ and $Y_{i}$, meaning the same random variable in different points in time. The last ones correspond to the solid arrows in Figure \ref{interactInt}. A graphical representation for $n=2$ can be found in the first column of Figure \ref{overall}.

\begin{Definition}[Stochastic Interaction]
The set of distributions belonging to the split model in the sense of Stochastic Interaction can be defined as 
\begin{equation*}
\mathcal{M}_{SI}=\left\lbrace Q \in \mathcal{P}(\mathcal{Z}) \mid Q(Y \mid X) = \bigotimes\limits_{i=1}^{n} Q(Y_{i} \mid X_{i}) \right\rbrace 
\end{equation*}
and the complexity measure can be calculated as follows
\begin{equation*}
\Phi_{SI}=\inf\limits_{Q \in M_{SI}} D_{\mathcal{Z}}(\tilde{P} \parallel Q)= \sum\limits_{i=1}^{n} H(Y_{i} \mid X_{i}) - H(Y \mid X),
\end{equation*}
as shown in Reference \cite{stochInt}. In the definition above, $H$ denotes the conditional entropy 
\begin{equation*}
H(Y_{i} \mid X_{i}) = - \sum\limits_{x_{i} \in \mathcal{X}_{i}} \sum\limits_{y_{i} \in \mathcal{Y}_{i}} \tilde{P}(x_{i},y_{i}) \, log \, \tilde{P}(y_{i} \vert x_{i}).
\end{equation*}
\end{Definition}
This does not satisfy Property \ref{postulate2} and therefore the corresponding graph is displayed only in the first column of Figure \ref{overall}.  {Amari points out in Reference \cite{amari} that this measure is not applicable in the case of an exterior influences on the $Y_{i}$s. Such an influence can cause the $Y_{i}$s to be correlated even in the case of independent $X_{i}$s and no causal cross-connections.} 

Consider a setting without exterior influences, then $\Phi_{SI}$ quantifies the strength of the causal cross-connections alone and is therefore a reasonable choice for an Integrated Information measure. Accounting for an exterior influence that does not exist leads to a split system, which compensates a part of the removal of the causal cross-connections so that the resulting measure does not quantify all of the interior causal cross-influences. 

To force the model to satisfy Property \ref{postulate2}, 
one can add the interaction between $Y_{i}$ and $Y_{j}$, which results in the measure Geometric Integrated Information \cite{amari}.  
\begin{Definition}[Geometric Integrated Information]
The graphical model corresponding to the graph in the second row and first column of Figure \ref{overall}
is the set 
\begin{equation*}
\mathcal{M}_{G}= \left\lbrace P \in \mathcal{P}(\mathcal{Z}) \vert \exists f_{1},\dots ,f_{n +2} \in \mathbb{R}_{+}^{\mathcal{Z}} \text{ s.t. } P(z) = f_{n+ 1}(x)f_{n+2}(y) \prod\limits_{i = 1}^{n} f_{i}(x_{i},y_{i}) \right\rbrace
\end{equation*} 
and the measure is defined as 
\begin{equation*}
\Phi_{G} = \inf\limits_{Q \in \mathcal{M}_{G}} D_{\mathcal{Z}}( \tilde{P} \parallel Q).
\end{equation*}
\end{Definition}
$\mathcal{M}_{G}$ is called the diagonally split model in Reference \cite{GeomInfInt}.
This is not causally split in the sense that the corresponding distributions in general do not satisfy Property \ref{postulate1}. It can be seen by analyzing the conditional independence structure of the graph as described in Appendix \ref{AppGraph}. By introducing the edges between the $Y_{i}$s as fixed, $\Phi_{G}$ might force these connections to be stronger than they originally are. A result of this might be that an effect of the causal cross-connections gets atoned for by the new edge. We discussed this above in the context of Property \ref{postulate2}. 

This measure has no closed form solution, but we are able to calculate the corresponding split system with the help of the iterative scaling algorithm (see, for example, Section 5.1 in Reference \cite{CsiszarShields}).

The first measure that satifies both properties is called "Integrated Information" \cite{uniframe}, its model is referred to by "Causally split model" in Reference \cite{GeomInfInt} and it is derived from the first property. Since we are able to define it using conditional independence statements, we will denote it by $\Phi_{CIS}$.  
It requires $Y_{i}$ to be independent of $X_{I \setminus \{i\}}$ given $X_{i}$. 
\begin{Definition}[Integrated Information]
The set of distributions, that belongs to the split system corresponding to integrated information, is defined as
\begin{equation}
\mathcal{M}_{CIS} = \left\lbrace Q \in \mathcal{P}(\mathcal{Z}) \mid Q(Y_{i} \mid X) = Q(Y_{i} \mid X_{i}) , \, \text{for all } i \in \{1, \dots , n\} \right\rbrace \label{Mg}
\end{equation}
and this leads to the measure
\begin{equation*}
\Phi_{CIS} = \inf\limits_{Q \in \mathcal{M}_{CIS}} D_{\mathcal{Z}}(\tilde{P} \parallel Q). 
\end{equation*}
\end{Definition}
We write the requirements to the distributions in (\ref{Mg}) as conditional independent statements
\begin{equation*}
Y_{i} \ci X_{I \setminus \{i\}} \mid X_{i}.
\end{equation*}
A detailed analysis of probabilistic independence statements can be found in Reference \cite{Milan}. 
Unfortunately, these conditional independence statements can not be encoded in terms of a chain graph in general. The definition of this measure arises naturally from Property \ref{postulate1} by applying the relation (\ref{MarkovPop}) 
\begin{equation*}
Q(X_{i},Y_{j} \mid X_{I \setminus \{i\}}) = Q(X_{i}\mid X_{I \setminus \{i\}})Q(Y_{j}\mid X_{I \setminus \{i\}}), \quad \, i \neq j 
\end{equation*}
to all pairs $i,j \in \{1, \dots , n \}$. This leads to 
\begin{equation}
Q(Y_{j} \vert X) = Q(Y_{j} \vert X_{j}) \label{CIS},
\end{equation}
as shown in Appendix \ref{AppProof}.

Note that this implies that every model satisfying Property \ref{postulate1} is a submodel of $\mathcal{M}_{CIS}$.
In order to show that $\Phi_{CIS}$ satisfies Property \ref{postulate1}, we are going to rewrite the condition in Property \ref{postulate1} as
\begin{equation*}
Q(Y_{j} \vert X) = Q(Y_{j} \vert X_{I \setminus \{i\}}).
\end{equation*}
The definition of $\mathcal{M}_{CIS}$ allows us to write 
\begin{equation*}
Q(Y_{j} \vert X) = Q(Y_{j} \vert X_{j}) = Q(Y_{j} \vert X_{I \setminus \{i\}}),
\end{equation*}
for $Q \in \mathcal{M}_{CIS}$.
Therefore $\Phi_{CIS}$ satisfies Property \ref{postulate1} and since $\mathcal{M}_{I}$ meets the conditional independence statements of Property \ref{postulate1} the relation $\mathcal{M}_{I} \subseteq \mathcal{M}_{CIS}$ holds and $\Phi_{CIS}$ fulfills Property \ref{postulate2}. 

In Reference \cite{uniframe} Oizumi et al.~derive an analytical solution for Gaussian variables, but there does not exist a closed form solution for discrete variables in general. Therefore they use Newton's method in the case of discrete variables. 

Due to the lack of a graphical representation, it is difficult to interpret the causal nature of the elements of $\mathcal{M}_{CIS}$. In Example \ref{ex} we will see a type of model that is part of $\mathcal{M}_{CIS}$, but which {has a graphical representation.  This model does not lie in the set of Markovian processes discussed in this article $MP(\mathcal{Z})$. Hence this implies that not all the split distributions in $\mathcal{M}_{CIS}$ arise from removing connections from a full distribution as depicted in Figure \ref{interactInt}. } 

\section{Causal Information Integration} \label{SectCII}
Inspired by the discussion about extrinsic and intrinsic influences in the context of Property \ref{postulate2}, we now utilize the notion of a common exterior influence to define the measure $\Phi_{CII}$, which we call Causal Information Integration. This measure should be used in case of an unknown exterior influence.  

\subsection{Definition} 
Explicitly including a common exterior influence allows us to avoid the problems of a fixed edge between the $Y_{i}$s discussed earlier. 
This leads to the graphs in Figure \ref{newapproach}.
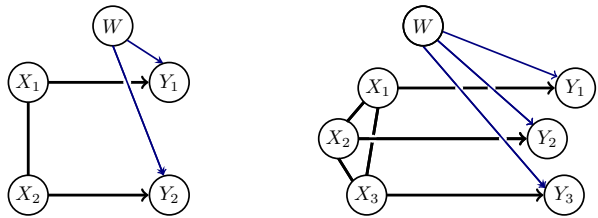
\begin{figure}[H]
\centering
\scalebox{0.75}{
\begin{tikzpicture}[rounded corners]
\draw[line width=0.5mm] (-1,0) node {$X_{1}$};
\draw[->,line width=0.5mm] (-0.65,0) coordinate (b_1)--(1.133,0) coordinate (b_2);
\draw[->,line width=0.5mm] (-0.65,-2)--(1.133,-2);
\draw[line width=0.3mm] (6,1) circle (10pt);
\draw[line width=0.3mm] (0.5,1) circle (10pt);
\draw[->,>=stealth,color =blue!50!black,line width=0.3mm] (0.75, 0.75) --(1.4,0.33);
\draw[->,color =blue!50!black,line width=0.3mm] (0.5, 0.65) coordinate (a_1) --(1.4,-1.65) coordinate (a_2);
\coordinate (k) at (intersection of b_1--b_2 and a_1--a_2);
\fill[white] (k) circle (2pt);
\draw[thick, line width=0.5mm] (5.2,-0.45) coordinate (g_1) --(5,-1.65) coordinate (g_2);
\draw[thick, line width=0.5mm] (4.95,-0.35)--(4.65,-0.7);
\draw[thick,line width=0.5mm] (4.75,-1.75)--(4.5,-1.35);
\draw[line width=0.5mm] (-1,-0.35)--(-1,-1.65);
\draw[->,>=stealth,color = blue!50!black,line width=0.3mm] (0.5, 0.65) coordinate (a_1) --(1.4,-1.65) coordinate (a_2);
\draw[line width=0.5mm] (0.5,1) node { $W$};
\draw[line width=0.5mm] (-1,-2) node {$X_{2}$};
\draw[line width=0.5mm] (1.5,0) node { $Y_{1}$};
\draw[line width=0.5mm] (1.5,-2) node { $Y_{2}$};
\draw[line width=0.3mm] (-1,0) circle (10pt);
\draw[line width=0.3mm] (-1,-2) circle (10pt);
\draw[line width=0.3mm] (1.5,0) circle (10pt);
\draw[line width=0.3mm] (1.5,-2) circle (10pt);

\draw[]  (5.2,-0.1) node (A) {$X_{1}$};
\draw[] (5,-2) node (B) {$X_{3}$};
\draw[] (4.5,-1) node (C) { $X_{2}$};
\draw[] (8.7,-0.1) node (D) {$Y_{1}$};
\draw[] (8.5,-2) node (E){$Y_{3}$};
\draw[] (8.2,-1) node (F) {$Y_{2}$};
\draw[] (6,1) node {$W$};
\draw[->,color =blue!50!black,line width=0.3mm] (6.25, 0.75) coordinate (a_1)--(7.95,-0.75) coordinate (a_2);
\draw[->,color = blue!50!black,line width=0.3mm] (6, 0.64) coordinate (b_1)--(8.166,-1.875) coordinate (b_2);
\draw[->,thick, line width=0.5mm] (5.55,-0.1) coordinate (c_1)--(8.333,-0.1) coordinate (c_2);
\draw[->,line width=0.5mm] (5.35,-2)--(8.133,-2);
\draw[->,line width=0.5mm] (4.85,-1) coordinate (h_1)--(7.833,-1) coordinate (h_2);
\draw[thick, line width=0.5mm] (5.2,-0.45) coordinate (g_1) --(5,-1.65) coordinate (g_2);
\draw[thick, line width=0.5mm] (4.95,-0.35)--(4.65,-0.7);
\draw[thick,line width=0.5mm] (4.75,-1.75)--(4.5,-1.35);
\coordinate (k) at (intersection of h_1--h_2 and g_1--g_2);
\fill[white] (k) circle (2pt);
\coordinate (j) at (intersection of a_1--a_2 and c_1--c_2);
\fill[white] (j) circle (2pt);
\coordinate (i) at (intersection of b_1--b_2 and c_1--c_2);
\fill[white] (i) circle (2pt);
\draw[->,line width=0.5mm] (4.85,-1) coordinate (h_1)--(7.833,-1) coordinate (h_2);
\coordinate (l) at (intersection of b_1--b_2 and h_1--h_2);
\fill[white] (l) circle (2pt);
\draw[->,color =blue!50!black,line width=0.3mm] (6.35, 0.9)--(8.38,0.1);
\draw[->,color = blue!50!black,line width=0.3mm] (6.25, 0.75) coordinate (a_1)--(7.95,-0.75) coordinate (a_2);
\draw[->,>=stealth,color = blue!50!black,line width=0.3mm] (6, 0.64) coordinate (b_1)--(8.166,-1.875) coordinate (b_2);

\draw[line width=0.3mm] (6,1) circle (10pt);
\draw[line width=0.3mm] (4.5,-1) circle (10pt);
\draw[line width=0.3mm] (5.2,-0.1) circle (10pt);
\draw[line width=0.3mm] (5,-2) circle (10pt);
\draw[line width=0.3mm] (8.7,-0.1) circle (10pt);
\draw[line width=0.3mm] (8.5,-2) circle (10pt);
\draw[line width=0.3mm] (8.2,-1) circle (10pt);
\end{tikzpicture} }
\caption{Split systems with exterior influences for $n=2$ and $n=3$. } \label{newapproach}
\end{figure}
The factorization of the distributions belonging to these graphical models is the following one
\begin{equation*}
P(z,w) = P(x) \prod\limits_{i=1}^{n}P(y_{i} \vert x_{i}, w)P(w).
\end{equation*} 
By marginalizing over the elements of $\mathcal{W}$ we get a distribution on $\mathcal{Z}$ defining our new model.

\begin{Definition}[Causal Information Integration] The set of distributions belonging to the marginalized model for $\vert \mathcal{W}^{m} \vert = m$ is 
\begin{equation*}
\mathcal{M}_{CII}^{m} = \left\lbrace P \in \mathcal{P}(\mathcal{Z}) \vert \exists Q \in \mathcal{P}(\mathcal{Z} \times \mathcal{W}^{m}) : P(z) = \sum\limits_{j = 1}^{m} Q(x) Q(w_{j}) \prod\limits_{i=1}^{n} Q(y_{i} \vert x_{i}, w_{j}) \right\rbrace. 
\end{equation*}
We will define the split model for Causal Integrated Information as the closure (denoted by a bar) of the union of $\mathcal{M}_{CII}^{m}$s:
\begin{equation}
\mathcal{M}_{CII} = \overline{\bigcup\limits_{m \in \mathbb{N}} \mathcal{M}_{CII}^{m}}. \label{setMCII}
\end{equation}
This leads to the measure
\begin{equation*}
\Phi_{CII} = \inf\limits_{Q \in \mathcal{M}_{CII}} D_{\mathcal{Z}}(\tilde{P} \parallel Q). 
\end{equation*}
\end{Definition}
{Since the split system $\mathcal{M}_{CII}$ was defined by utilizing graphs, we are able to use the graphical representation to get a more precise notion of the cases in which $\Phi_{CII}(\tilde{P})=0$ holds. In those cases the initial distribution can be completely explained as a limit of marginalized distributions without causal cross-influences and with exterior influences.
\begin{Proposition} \label{IfandOnlyIf}
The measure $\Phi_{CII}(\tilde{P})$ is 0 if and only if there exists a sequence of distributions $Q^{m} \in \mathcal{P}(\mathcal{Z})$ with the following properties.
\begin{itemize}
\item[1.] $ \tilde{P} = \lim\limits_{m \rightarrow \infty} Q^{m}.$
\item[2.] For every $m \in \mathbb{N}$ there exists a distribution $\hat{Q}^{m} \in \mathcal{P}(\mathcal{Z} \times \mathcal{W}^{m})$ that has $\mathcal{Z}$ marginals equal to $Q^{m}$
\begin{equation*}
Q^{m}(z) = \hat{Q}^{m}(z), \quad \forall z \in \mathcal{Z}.
\end{equation*}
Additionally  $\hat{Q}^{m}$ factors according to the graph corresponding to the split system
\begin{equation*}
\hat{Q}^{m}(z,w) = \hat{Q}(x)^{m} \prod\limits_{i=1}^{n} \hat{Q}^{m}(y_{i} \vert x_{i}, w)\hat{Q}^{m}(w), \quad \forall (z,w) \in \mathcal{Z} \times \mathcal{W}^{m}.
\end{equation*}
\end{itemize}
\end{Proposition}  
}

In order to show that $\Phi_{CII}$ satisfies the conditional independence statements in Property \ref{postulate1}, we will calculate the conditional distributions $P(y_{i} \vert x_{i})$ and $P(y_{i} \vert x)$ of 
\begin{equation*}
P(z) = \sum\limits_{w} P(x)\prod\limits_{j=1}^{n}P(y_{j} \vert x_{j}, w)P(w).  
\end{equation*}
This results in
\begin{equation*}
\begin{split}
P(y_{i} \vert x_{i}) &=  \dfrac{\sum\limits_{y_{I \setminus \{i\}}}\sum\limits_{x_{I \setminus \{i\}}} \sum\limits_{w}P(x)\prod\limits_{i=j}^{n}P(y_{j} \vert x_{j}, w)P(w) }{P(x_{i})} =  \dfrac{\sum\limits_{x_{I \setminus \{i\}}} \sum\limits_{w}P(x)P(y_{i} \vert x_{i}, w)P(w) }{P(x_{i})} =  \sum\limits_{w} P(y_{i} \vert x_{i}, w)P(w)  \\  
P(y_{i} \vert x) &= \dfrac{\sum\limits_{y_{I \setminus \{i\}}}\sum\limits_{w}P(x)\prod\limits_{i=j}^{n}P(y_{j} \vert x_{j}, w)P(w) }{P(x)} =  \sum\limits_{w} P(y_{i} \vert x_{i}, w)P(w)  \\ 
\end{split} 
\end{equation*}
for all $z \in \mathcal{Z}$. Hence $P(y_{i} \vert x_{i}) = P(y_{i} \vert x)$, for every $P \in \mathcal{M}_{CII}^{m}, \, m \in \mathbb{N}$. Since every element in $\hat{P} \in \mathcal{M}_{CII}$ is a limit point of distributions that satisfy the conditional independence statements, $\hat{P}$ also fulfills those. A proof can be found in Reference \cite{graphModels} Proposition 3.12. Therefore $\Phi_{CII}$ satisfies Property \ref{postulate1} and the set of all such distributions is a subset of $\mathcal{M}_{CIS}$
\begin{equation*}
\mathcal{M}_{CII}  \subseteq \mathcal{M}_{CIS}.
\end{equation*}

We are able to represent the marginalized model by using the methods from Reference \cite{Marginal}. Up to this point we have been using chain graphs. These are graphs consisting of directed and undirected edges such that there are no semi-directed cycles as described in Appendix \ref{AppGraph}. In order to be able to gain a graph that represents the conditional independence structure of the marginalized model, we need the concept of chain mixed graphs (CMGs). In addition to the directed and undirected edges belonging to chain graphs, chain mixed graphs also have arcs $\leftrightarrow$. Two nodes connected by an arc are called spouses. The connection between spouses appears when we marginalize over a common influence, hence spouses do not have a directed information flow from one node to the other but are affected by the same mechanisms.
The Algorithm \ref{Algorithm} from Reference \cite{Marginal}  allows us to transform a chain graph with latent variables into a chain mixed graph that represents the conditional independence structures of the marginalized chain graph. Using this on the graphs in Figure \ref{newapproach} leads to the CMGs in Figure \ref{MarginalizedModel}. Unfortunately, there exists no new factorization corresponding to the CMGs known to the authors.

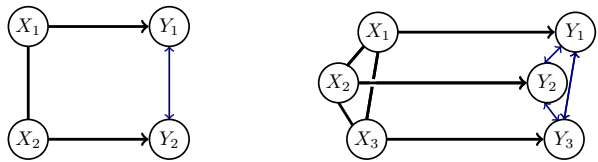
\begin{figure}[H]
\centering
\scalebox{0.75}{
\begin{tikzpicture}[rounded corners]
\draw[line width=0.5mm] (-1,0) node {$X_{1}$};
\draw[->,line width=0.5mm] (-0.65,0) coordinate (b_1)--(1.133,0) coordinate (b_2);
\draw[->,line width=0.5mm] (-0.65,-2)--(1.133,-2);
\draw[thick, line width=0.5mm] (5.2,-0.45) coordinate (g_1) --(5,-1.65) coordinate (g_2);
\draw[thick, line width=0.5mm] (4.95,-0.35)--(4.65,-0.7);
\draw[thick,line width=0.5mm] (4.75,-1.75)--(4.5,-1.35);
\draw[line width=0.5mm] (-1,-0.35)--(-1,-1.65);
\draw[line width=0.3mm,blue!50!black,<->] (1.5,-0.35)--(1.5,-1.65);
\draw[line width=0.5mm] (-1,-2) node {$X_{2}$};
\draw[line width=0.5mm] (1.5,0) node { $Y_{1}$};
\draw[line width=0.5mm] (1.5,-2) node { $Y_{2}$};
\draw[line width=0.3mm] (-1,0) circle (10pt);
\draw[line width=0.3mm] (-1,-2) circle (10pt);
\draw[line width=0.3mm] (1.5,0) circle (10pt);
\draw[line width=0.3mm] (1.5,-2) circle (10pt);

\draw[]  (5.2,-0.1) node (A) {$X_{1}$};
\draw[] (5,-2) node (B) {$X_{3}$};
\draw[] (4.5,-1) node (C) { $X_{2}$};
\draw[] (8.7,-0.1) node (D) {$Y_{1}$};
\draw[] (8.5,-2) node (E){$Y_{3}$};
\draw[] (8.2,-1) node (F) {$Y_{2}$};
\draw[->,thick, line width=0.5mm] (5.55,-0.1) coordinate (c_1)--(8.333,-0.1) coordinate (c_2);
\draw[->,line width=0.5mm] (5.35,-2)--(8.133,-2);
\draw[->,line width=0.5mm] (4.85,-1) coordinate (h_1)--(7.833,-1) coordinate (h_2);
\draw[thick, line width=0.5mm] (5.2,-0.45) coordinate (g_1) --(5,-1.65) coordinate (g_2);
\draw[line width=0.3mm, blue!50!black,<->] (8.5,-1.65) --(8.7,-0.45);
\draw[thick, line width=0.5mm] (4.95,-0.35)--(4.65,-0.7);
\draw[thick,line width=0.3mm, blue!50!black, <->]  (8.45,-0.35)--(8.15,-0.65);
\draw[thick,line width=0.3mm, blue!50!black,<->] (8.5,-1.65) --(8.7,-0.45);
\draw[thick, line width=0.3mm, blue!50!black, <->] (8.15,-1.35)--(8.4,-1.68);
\coordinate (k) at (intersection of h_1--h_2 and g_1--g_2);
\fill[white] (k) circle (2pt);
\draw[->,line width=0.5mm] (4.85,-1) coordinate (h_1)--(7.833,-1) coordinate (h_2);

\draw[line width=0.3mm] (4.5,-1) circle (10pt);
\draw[line width=0.3mm] (5.2,-0.1) circle (10pt);
\draw[line width=0.3mm] (5,-2) circle (10pt);
\draw[line width=0.3mm] (8.7,-0.1) circle (10pt);
\draw[line width=0.3mm] (8.5,-2) circle (10pt);
\draw[line width=0.3mm] (8.2,-1) circle (10pt);
\end{tikzpicture} }
\caption{Marginalized Model for $n=2$ and $n=4$. } \label{MarginalizedModel}
\end{figure}

In order to prove that $ \Phi_{CII}$ satisfies Property \ref{postulate2}, we will show that $\mathcal{M}_{I}$ is a  subset of $\mathcal{M}_{CII}$. At first we will consider the following subset of $\mathcal{M}_{CII}$
\begin{equation*}
\begin{split}
\mathcal{M}_{CI}^{m} &= \left\lbrace P \in \mathcal{P}(\mathcal{Z}) \vert \exists Q \in \mathcal{P}(\mathcal{Z} \times \mathcal{W}^{m}) : P(z) = \sum\limits_{j = 1}^{m} Q(x) Q(w_{j})\prod\limits_{i=1}^{n} Q(y_{i} \vert w_{j})   \right\rbrace \\
\mathcal{M}_{CI} &= \overline{\bigcup\limits_{m \in \mathbb{N}} \mathcal{M}_{CI}^{m}},
\end{split}
\end{equation*} where we remove the connections between the different stages, as shown in Figure \ref{I}. 

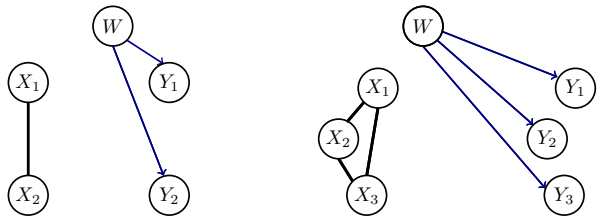
\begin{figure}[H]
\centering
\scalebox{0.75}{
\begin{tikzpicture}[rounded corners]
\draw[line width=0.5mm] (-1,0) node {$X_{1}$};
\draw[line width=0.3mm] (6,1) circle (10pt);
\draw[line width=0.3mm] (0.5,1) circle (10pt);
\draw[->,color = blue!50!black,line width=0.3mm] (0.75, 0.75) --(1.4,0.33);
\draw[->,color = blue!50!black,line width=0.3mm] (0.5, 0.65) coordinate (a_1) --(1.4,-1.65) coordinate (a_2);
\draw[thick, line width=0.5mm] (5.2,-0.45) coordinate (g_1) --(5,-1.65) coordinate (g_2);
\draw[thick, line width=0.5mm] (4.95,-0.35)--(4.65,-0.7);
\draw[thick,line width=0.5mm] (4.75,-1.75)--(4.5,-1.35);
\draw[line width=0.5mm] (-1,-0.35)--(-1,-1.65);
\draw[->,color = blue!50!black,line width=0.3mm] (0.5, 0.65) coordinate (a_1) --(1.4,-1.65) coordinate (a_2);
\draw[line width=0.5mm] (0.5,1) node { $W$};
\draw[line width=0.5mm] (-1,-2) node {$X_{2}$};
\draw[line width=0.5mm] (1.5,0) node { $Y_{1}$};
\draw[line width=0.5mm] (1.5,-2) node { $Y_{2}$};
\draw[line width=0.3mm] (-1,0) circle (10pt);
\draw[line width=0.3mm] (-1,-2) circle (10pt);
\draw[line width=0.3mm] (1.5,0) circle (10pt);
\draw[line width=0.3mm] (1.5,-2) circle (10pt);

\draw[]  (5.2,-0.1) node (A) {$X_{1}$};
\draw[] (5,-2) node (B) {$X_{3}$};
\draw[] (4.5,-1) node (C) { $X_{2}$};
\draw[] (8.7,-0.1) node (D) {$Y_{1}$};
\draw[] (8.5,-2) node (E){$Y_{3}$};
\draw[] (8.2,-1) node (F) {$Y_{2}$};
\draw[] (6,1) node {$W$};
\draw[->,color = blue!50!black,line width=0.3mm] (6.35, 0.9)--(8.38,0.1);
\draw[->,color = blue!50!black,line width=0.3mm] (6.25, 0.75) coordinate (a_1)--(7.95,-0.75) coordinate (a_2);
\draw[->,color = blue!50!black,line width=0.3mm] (6, 0.64) coordinate (b_1)--(8.166,-1.875) coordinate (b_2);
\draw[thick, line width=0.5mm] (5.2,-0.45) coordinate (g_1) --(5,-1.65) coordinate (g_2);
\draw[thick, line width=0.5mm] (4.95,-0.35)--(4.65,-0.7);
\draw[thick,line width=0.5mm] (4.75,-1.75)--(4.5,-1.35);
\coordinate (l) at (intersection of b_1--b_2 and h_1--h_2);
\draw[->,color = blue!50!black,line width=0.3mm] (6.35, 0.9)--(8.38,0.1);
\draw[->,color = blue!50!black,line width=0.3mm] (6.25, 0.75) coordinate (a_1)--(7.95,-0.75) coordinate (a_2);
\draw[->,color = blue!50!black,line width=0.3mm] (6, 0.64) coordinate (b_1)--(8.166,-1.875) coordinate (b_2);

\draw[line width=0.3mm] (6,1) circle (10pt);
\draw[line width=0.3mm] (4.5,-1) circle (10pt);
\draw[line width=0.3mm] (5.2,-0.1) circle (10pt);
\draw[line width=0.3mm] (5,-2) circle (10pt);
\draw[line width=0.3mm] (8.7,-0.1) circle (10pt);
\draw[line width=0.3mm] (8.5,-2) circle (10pt);
\draw[line width=0.3mm] (8.2,-1) circle (10pt);
\end{tikzpicture} }
\caption{Submodels of the split models with exterior influences for $n=2$ and $n=3$. }  \label{I}
\end{figure}
Now $X$ and $Y$ are independent of each other
\begin{equation*}
Q(z) = Q(x) \cdot Q(y)
\end{equation*} 
with 
\begin{equation*}
Q(y) =  \sum\limits_{w} Q(w) \prod\limits_{i=1}^{n} Q(y_{i} \vert w)
\end{equation*}
for $Q \in \mathcal{M}_{CI}^{m}$ and since independence structures of discrete distributions are preserved in the limit we have $\mathcal{M}_{CI} \subseteq \mathcal{M}_{I}$. In order to gain equality it remains to show that $Q(Y)$ can approximate every distribution on $\mathcal{Y}$ if the state space of $W$ is sufficiently large. These distributions are mixtures of discrete product distributions, where \begin{equation*}
\prod\limits_{i=1}^{n} Q(y_{i} \vert w)
\end{equation*} 
are the mixture components and $Q(w)$ are the mixture weights. Hence we are able to use the following result.

\begin{Theorem}[Theorem 1.3.1 from Reference \cite{Guido}] \label{ThmGuido}
Let $q$ be a prime power. The smallest $m$ for which any probability distribution on $\{1, \dots ,q\}$ can be approximated arbitrarily well as mixture of $m$ product distributions is $q^{n-1}$.
\end{Theorem}
Universal approximation results like the theorem above may suggest that the models $\mathcal{M}_{CII}$ and $\mathcal{M}_{CIS}$ are equal. However we will present numerically calculated examples of elements belonging to $\mathcal{M}_{CIS}$, but not to $\mathcal{M}_{CII}$, even with an extremely large state space. We will discuss this matter further in Section \ref{SectRel}.

In conclusion, $\Phi_{CII}$ satisfies Property \ref{postulate1} and \ref{postulate2}.

 {Note that using $\Phi_{CII}$ in cases without an exterior influence might not capture all the internal cross-influences, since the additional latent variable can compensate some of the difference between the initial distribution and the split model. This can only be avoided when the exterior influence is known and can therefore be included in the model. We will discuss that case in the next section. }

\setcounter{subsection}{1}
\subsubsection{Ground Truth} \label{Sectgroundtruth}
{The} concept of an exterior influence suggests that there exists a ground truth in a larger model in which $W$ is a visible variable. This is shown in Figure \ref{groundtruth} on the right.
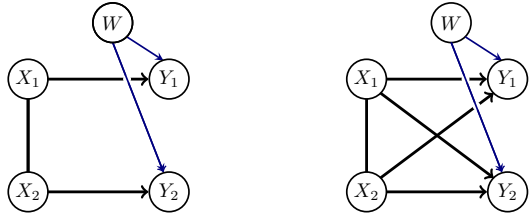
\begin{figure}[H]
\centering
\scalebox{0.75}{
\begin{tikzpicture}[rounded corners]
\draw[line width=0.5mm] (-1,0) node {$X_{1}$};
\draw[->,line width=0.5mm] (-0.65,0) coordinate (b_1)--(1.133,0) coordinate (b_2);
\draw[->,line width=0.5mm] (-0.65,-2)--(1.133,-2);
\draw[line width=0.3mm] (0.5,1) circle (10pt);
\draw[->,>=stealth,color =blue!50!black,line width=0.3mm] (0.75, 0.75) --(1.4,0.33);
\draw[->,color =blue!50!black,line width=0.3mm] (0.5, 0.65) coordinate (a_1) --(1.4,-1.65) coordinate (a_2);
\coordinate (k) at (intersection of b_1--b_2 and a_1--a_2);
\fill[white] (k) circle (2pt);
\draw[line width=0.5mm] (-1,-0.35)--(-1,-1.65);
\draw[->,>=stealth,color = blue!50!black,line width=0.3mm] (0.5, 0.65) coordinate (a_1) --(1.4,-1.65) coordinate (a_2);
\draw[line width=0.5mm] (0.5,1) node { $W$};
\draw[line width=0.5mm] (-1,-2) node {$X_{2}$};
\draw[line width=0.5mm] (1.5,0) node { $Y_{1}$};
\draw[line width=0.5mm] (1.5,-2) node { $Y_{2}$};
\draw[line width=0.3mm] (-1,0) circle (10pt);
\draw[line width=0.3mm] (-1,-2) circle (10pt);
\draw[line width=0.3mm] (1.5,0) circle (10pt);
\draw[line width=0.3mm] (1.5,-2) circle (10pt);

\draw[line width=0.5mm] (5,0) node {$X_{1}$};
\draw[->,line width=0.5mm] (5.35,0) coordinate (b_1)--(7.133,0) coordinate (b_2);
\draw[->,line width=0.5mm] (5.35,-2)--(7.133,-2);
\draw[line width=0.3mm] (6.5,1) circle (10pt);
\draw[line width=0.3mm] (0.5,1) circle (10pt);
\draw[->,>=stealth,color =blue!50!black,line width=0.3mm] (6.75, 0.75) --(7.4,0.33);
\draw[->,color =blue!50!black,line width=0.3mm] (6.5, 0.65) coordinate (a_1) --(7.4,-1.65) coordinate (a_2);
\coordinate (k) at (intersection of b_1--b_2 and a_1--a_2);
\fill[white] (k) circle (2pt);
\draw[thick, line width=0.5mm] (5,-0.35) coordinate (g_1) --(5,-1.65) coordinate (g_2);
\draw[line width=0.5mm] (-1,-0.35)--(-1,-1.65);
\draw[->, line width = 0.5mm] (5.25, -1.75)--(7.25, -0.25);
\draw[->, line width = 0.5mm] (5.25, -0.25)--(7.25, -1.75);
\fill[white] (6.95, -0.48) circle (2pt);
\draw[->,>=stealth,color = blue!50!black,line width=0.3mm] (6.5, 0.65) coordinate (a_1) --(7.4,-1.65) coordinate (a_2);
\draw[line width=0.5mm] (6.5,1) node { $W$};
\draw[line width=0.5mm] (5,-2) node {$X_{2}$};
\draw[line width=0.5mm] (7.5,0) node { $Y_{1}$};
\draw[line width=0.5mm] (7.5,-2) node { $Y_{2}$};
\draw[line width=0.3mm] (5,0) circle (10pt);
\draw[line width=0.3mm] (5,-2) circle (10pt);
\draw[line width=0.3mm] (7.5,0) circle (10pt);
\draw[line width=0.3mm] (7.5,-2) circle (10pt);
\end{tikzpicture} }
\caption{The graphs corresponding to $\mathcal{E}$ and $\mathcal{E}^{f}$ (right).} \label{groundtruth}
\end{figure}
Assuming that we know the distribution of the whole model, we are able to apply the concepts discussed above to define an Integrated Information measure $\Phi_{T}$ on the larger space. This allows us to really only remove the causal cross-connections as shown in Figure \ref{groundtruth} on the left. Thus we can interpret $\Phi_{T}$ as the ultimate measure of Integrated Information, if the ground truth is available.
{Note that using the measure $\Phi_{SI}$ in the setting with no external influences is a special case of $\Phi_{T}$.}

The set of distributions belonging to the larger, fully connected model will be called $\mathcal{E}^{f}$ and the set corresponding to the graph on the left of Figure \ref{groundtruth} depicts the split system which will be denoted by $\mathcal{E}$. Since $W$ is now known, we are able to fix the state space $\mathcal{W}$ to its actual size $m$.
\begin{equation*}
\begin{split}
\mathcal{E} = \left\lbrace P \in \mathcal{P}(\mathcal{Z} \times \mathcal{W}^{m}) \mid P(z,w) = P(x) \prod\limits_{i=1}^{n}P(y_{i} \vert x_{i}, w)P(w), \, \forall (z,w) \in \mathcal{Z} \times \mathcal{W}^{m}, \, \vert \mathcal{W} \vert = m \right\rbrace \\
\mathcal{E}^{f} = \left\lbrace P \in \mathcal{P}(\mathcal{Z} \times \mathcal{W}^{m}) \mid P(z,w) = P(x) \prod\limits_{i=1}^{n}P(y_{i} \vert x, w)P(w), \, \forall (z,w) \in \mathcal{Z} \times \mathcal{W}^{m}, \, \vert \mathcal{W} \vert = m \right\rbrace.
\end{split}
\end{equation*}
Note that $\mathcal{E}$ is the set of all the distributions that result in an element of $\mathcal{M}_{CII}$ after marginalization over $\mathcal{W}^{m}$
\begin{equation*}
\mathcal{M}_{CII}^{m} = \left\lbrace P \in \mathcal{P}(\mathcal{Z}) \vert \exists Q \in \mathcal{E}^{m} : P(z) = \sum\limits_{j = 1 }^{m} Q(x) Q(w_{j}) \prod\limits_{i=1}^{n} Q(y_{i} \vert x_{i}, w_{j}) \right\rbrace.
\end{equation*}
Calculating the KL-divergence between $P \in \mathcal{E}^{f}$ and $\mathcal{E}$ results in the new measure.

\begin{Proposition} \label{PropGroundTruth}
Let $P \in \mathcal{E}^{f}$. Minimizing the KL-divergence between $P$ and $\mathcal{E}$ leads to
\begin{equation*}
\begin{split}
\Phi_{T} = \inf\limits_{Q \in \mathcal{E}} D_{\mathcal{Z} \times \mathcal{W}^{m}} (P \parallel Q) &= \sum\limits_{z,w} P(z,w) \, log \, \dfrac{\prod\limits_{i} P(y_{i} \vert x,w)}{\prod\limits_{i} P(y_{i} \vert x_{i}, w)} \\
 &= \sum\limits_{i} I(Y_{i} ; X_{I \setminus \{i\}} \vert X_{i}, W).
\end{split}
\end{equation*}
\end{Proposition}
In the definition above $ I(Y_{i} ; X_{I \setminus \{i\}} \vert X_{i}, W)$ is the conditional mutual information defined by
\begin{equation*}
I(Y_{i} ; X_{I \setminus \{i\}} \vert X_{i}, W)= \sum\limits_{y_{i}, x ,w} P(y_{i},x ,w) \, log \, \dfrac{P(y_{i},x_{I \setminus \{i\}} \vert x_{i}, w)}{P(y_{i} \vert x_{i}, w)P(x_{I \setminus \{i\}} \vert x_{i}, w)}.
\end{equation*}
It characterizes the reduction of uncertainty in $Y_{i}$ due to $X_{I\setminus \{i\}}$ when $W$ and $X_{i}$ are given.
Therefore this measure decomposes to a sum in which each addend characterizes the information flow towards one $Y_{i}$.
Writing this as conditional independence statements, $\Phi_{T}$ is 0 if and only if 
\begin{equation*}
Y_{i} \ci X_{I \setminus \{i\}} \vert \{X_{i}, W\}.
\end{equation*}
Ignoring $W$ would lead exactly to the conditional independence statements in equation (\ref{Mg}). For a more detailed description of the conditional mutual information and its properties, see Reference \cite{Wiley}.
 
{Furthermore, $\Phi_{T} = 0$ if and only if the initial distribution $P$ factors according to the graph that belongs to $\mathcal{E}$. This follows from Proposition \ref{PropGroundTruth} and the fact that the KL-divergence is 0 if and only if both distributions are equal. Hence this measure truly removes the causal cross-connections. }

Additionally, by using that $W \ci X$, we are able to split up the conditional mutual information into a part corresponding to the conditional independence statements of Property \ref{postulate1} and another conditional mutual information. 
\begin{equation*}
\begin{split}
I(Y_{i} ; X_{I \setminus \{i\}} \vert X_{i}, W) &= \sum\limits_{y_{i}, x,w} P(w) \, log \, \left( \dfrac{P(y_{i},x_{I \setminus \{i\}} \vert x_{i}) }{P(y_{i} \vert x_{i}) P(x_{I \setminus \{i\}} \vert x_{i})} \cdot  \dfrac{P(y_{i},x_{i})P(x)P(y_{i},x,w) P(x_{i},w)}{P(y_{i}, x)P(x_{i})P(y_{i}, x_{i}, w) P(x,w)}\right) \\
&= I(Y_{i}; X_{I \setminus \{i\}} \vert X_{i}) + \sum\limits_{y_{i}, x,w} P(w) \, log \,  \dfrac{P(y_{i},x_{i})P(x)P(y_{i},x,w) P(x_{i},w)}{P(y_{i}, x)P(x_{i})P(y_{i}, x_{i}, w) P(x,w)}  \\
&= I(Y_{i}; X_{I \setminus \{i\}} \vert X_{i}) + \sum\limits_{y_{i}, x,w} P(w) \, log \,  \dfrac{P(w,x_{I \setminus \{i\}} \vert y_{i}, x_{i})}{P(w \vert y_{i},x_{i})P(x_{I \setminus \{i\}} \vert y_{i}, x_{i})}  \\
&= I(Y_{i}; X_{I \setminus \{i\}} \vert X_{i}) + I(W ; X_{I \setminus \{i\}} \vert Y_{i}, X_{i}).
\end{split}
\end{equation*}
Since the conditional mutual information is non-negative, $\Phi_{T}$ is 0 if and only if the conditional independence statements of equation (\ref{Mg}) hold and additionally the reduction of uncertainty in $W$ due to $X_{I \setminus \{i\}}$ given $Y_{i}, X_{i}$ is 0.

In general, we do not know what the ground truth of our system is and therefore we have to assume that $W$ is a hidden variable. This leads us back to $\Phi_{CII}$. Minimizing over all possible $W$ might compensate a part of the causal information flow. {One example, in which accounting for an exterior influence that does not exist leads to a value smaller than the true integrated information, was discussed earlier in the context of Property \ref{postulate2}. There we refer to an example in Reference \cite{Comparison} where $\Phi_{SI}$ exceeds $\Phi_{I}$ in a setting without an exterior influence. Similarly, $\Phi_{CII}$ is smaller or equal to the true value $\Phi_{T}$. } 

\begin{Proposition} \label{relGroundTruth}
The new measure $\Phi_{T}$ is an upper bound for $\Phi_{CII}$
\begin{equation*}
\Phi_{CII} \leq \Phi_{T}.
\end{equation*}
\end{Proposition}

Hence by assuming that there exists a common exterior influence, we are able to show that $\Phi_{CII}$ is bounded from above by the true value, that measures all the intrinsic cross-influences. {We are able to observe this behavior in Section \ref{sectResults}}.    

\subsubsection{Relationships between the different measures} \label{SectRel}
Now we are going to analyze the relationship between the different measures $\Phi_{SI}, \Phi_{G}, \Phi_{CIS}$ and $\Phi_{CII}$. We will start with $\Phi_{G}$ and $\Phi_{CII}$. 
Previously we already showed that $ \Phi_{CII}$ satisfies Property \ref{postulate1} and since $\Phi_{G}$ does not satisfy Property \ref{postulate1}, we have 
\begin{equation*}
\mathcal{M}_{G} \nsubseteq \mathcal{M}_{CII}.
\end{equation*}
To evaluate the other inclusion, we will consider
the more refined parametrizations of elements $P \in \mathcal{M}_{CII}^{m}$ and $Q \in \mathcal{M}_{G}$ as defined \ref{factor}. These are
\begin{equation*}
\begin{split}
P(z) &= P(x)f_{2}(x_{1},y_{1})g_{2}(x_{2},y_{2}) \sum\limits_{w } P(w) f_{1}(w,y_{1})f_{3}(x_{1},y_{1},w)g_{1}(w,y_{2})g_{3}(x_{2},y_{2},w) \\
 &=  P(x)f_{2}(x_{1},y_{1})g_{2}(x_{2},y_{2})  \phi (x_{1}, x_{2}, y_{1}, y_{2}) \\
Q(z) &= h_{n+1}(x)h_{n+2}(y) \prod\limits_{i =1}^{n}h_{i}(y_{i},x_{i}),
\end{split}
\end{equation*}
where $f_{1},f_{2},f_{3},g_{1},g_{2},g_{3},h_{1},h_{2},h_{3},h_{4}$ are non-negative functions such that $P,Q \in \mathcal{P}(\mathcal{Z})$ and 
\begin{equation*}
\phi (x_{1}, x_{2}, y_{1}, y_{2}) =  \sum\limits_{w } P(w) f_{1}(w,y_{1})f_{3}(x_{1},y_{1},w)g_{1}(w,y_{2})g_{3}(x_{2},y_{2},w).
\end{equation*}  Since $\phi$ depends on more than $Y_{1}$ and $Y_{2}$, $P(z)$ does not factorize according to $\mathcal{M}_{G}$ in general. Hence $\mathcal{M}_{CII} \nsubseteq \mathcal{M}_{G}$ holds. 

Furthermore, looking at the parametrizations allows us to identify a subset of distributions that lies in the intersection of $\mathcal{M}_{G}$ and $\mathcal{M}_{CII}$. Allowing $P$ to only have pairwise interactions would lead to 
\begin{equation*}
\begin{split}
P(z) &= P(x)\tilde{f}_{2}(x_{1},y_{1})\tilde{g}_{2}(x_{2},y_{2}) \sum\limits_{w } P(w) \tilde{f}_{1}(w,y_{1})\tilde{g}_{1}(w,y_{2}) \\
&=P(x)\tilde{f}_{2}(x_{1},y_{1})\tilde{g}_{2}(x_{2},y_{2}) \tilde{\phi}(y_{1}, y_{2}),
\end{split}
\end{equation*}
with the non-negative functions $\tilde{f}_{1},\tilde{f}_{2}, \tilde{g}_{1},\tilde{g}_{2}$ such that $P \in \mathcal{P}(\mathcal{Z})$ and 
\begin{equation*}
\tilde{\phi}(y_{1}, y_{2}) = \sum\limits_{w } P(w) \tilde{f}_{1}(w,y_{1})\tilde{g}_{1}(w,y_{2}).
\end{equation*}
This $P$ is an element of $\mathcal{M}_{G} \cap \mathcal{M}_{CII}$.

In the next part we will discuss the relationship between $\mathcal{M}_{CII}$ and $\mathcal{M}_{CIS}$. The elements in $\mathcal{M}_{CII}$ satisfy the conditional independence statements of Property \ref{postulate1}, therefore
\begin{equation*}
\mathcal{M}_{CII} \subseteq \mathcal{M}_{CIS}.
\end{equation*}
Previously we have seen that making the state space of $W$ large enough can approximate a distribution between the $Y_{i}$s, see Theorem \ref{ThmGuido}. This gives the impression that $\mathcal{M}_{CII}$ and $\mathcal{M}_{CIS}$ coincide. However, based on numerically calculated examples, we have the following conjecture.
\begin{Conjecture} \label{Conj} It is not possible to approximate every distribution $Q \in \mathcal{M}_{CIS}$ with arbitrary accuracy by an element of $P \in \mathcal{M}_{CII}$. Therefore, we have that
\begin{equation*}
\mathcal{M}_{CII} \subsetneq \mathcal{M}_{CIS}.
\end{equation*}
\end{Conjecture} 
The following example strongly suggests this conjecture to be true.
\begin{Example} \label{ex}
Consider the set of distributions that factor according to the graph in Figure \ref{small1}
\begin{equation*}
\mathcal{N}_{CIS} =  \{P \in \mathcal{P}(\mathcal{Z}) \vert P(z) = P(x_{1})P(x_{2})P(y_{1} \vert x_{1}, y_{2})P(y_{2}) \}.
\end{equation*}
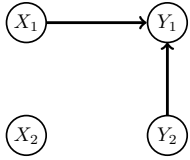
\begin{figure}[H]
\centering
\scalebox{0.75}{
\begin{tikzpicture}[rounded corners]
\draw[line width=0.5mm] (0,0) node {$X_{1}$};
\draw[->,line width=0.5mm] (0.35,0)--(2.133,0);
\draw[<-, line width=0.5mm] (2.5,-0.35)--(2.5,-1.65);
\draw[line width=0.5mm] (0,-2) node {$X_{2}$};
\draw[line width=0.5mm] (2.5,0) node { $Y_{1}$};
\draw[line width=0.5mm] (2.5,-2) node { $Y_{2}$};
\draw[line width=0.3mm] (0,0) circle (10pt);
\draw[line width=0.3mm] (0,-2) circle (10pt);
\draw[line width=0.3mm] (2.5,0) circle (10pt);
\draw[line width=0.3mm] (2.5,-2) circle (10pt);
\end{tikzpicture}
}
\caption{Graph of the model $\mathcal{N}_{CIS}$.} \label{small1}
\end{figure} 
This model satisfies the conditional independence statements of Property \ref{postulate1} and is therefore a subset of the model $\mathcal{M}_{CIS}$. In this case $X_{1}$ and $X_{2}$ are independent of each other, hence from a causal perspective the influence of $Y_{2}$ on $Y_{1}$ should be purely external. Therefore we try to model this with a subset of $\mathcal{M}_{CII}$ 
\begin{equation}
\begin{split}
\mathcal{N}_{CII} &= \overline{\bigcup\limits_{m \in \mathbb{N}} \mathcal{N}^{m}_{CII}}, \\
\mathcal{N}_{CII}^{m} &= \left\lbrace P \in \mathcal{P}(\mathcal{Z}) \vert \exists Q \in \mathcal{P}(\mathcal{Z} \times \mathcal{W}^{m}) : P(z) =Q(x_{1})Q(x_{2})\sum\limits_{j = 1}^{m} Q(y_{1} \vert x_{1}, w_{j})Q(y_{2} \vert w_{j}) Q(w_{j}) \right\rbrace
\end{split}
\end{equation}
and this corresponds to Figure \ref{small2}.
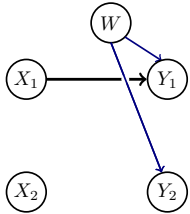
\begin{figure}[H]
\centering
\scalebox{0.75}{
\begin{tikzpicture}[rounded corners]
\draw[line width=0.5mm] (-1,0) node {$X_{1}$};
\draw[->,line width=0.5mm] (-0.65,0) coordinate (b_1)--(1.133,0) coordinate (b_2);
\draw[line width=0.3mm] (0.5,1) circle (10pt);
\draw[->,color = blue!50!black,line width=0.3mm] (0.75, 0.75) --(1.4,0.33);
\draw[->,color = blue!50!black,line width=0.3mm] (0.5, 0.65) coordinate (a_1) --(1.4,-1.65) coordinate (a_2);
\coordinate (k) at (intersection of b_1--b_2 and a_1--a_2);
\fill[white] (k) circle (2pt);
\draw[line width=0.5mm] (0.5,1) node { $W$};
\draw[line width=0.5mm] (-1,-2) node {$X_{2}$};
\draw[line width=0.5mm] (1.5,0) node { $Y_{1}$};
\draw[line width=0.5mm] (1.5,-2) node { $Y_{2}$};
\draw[line width=0.3mm] (-1,0) circle (10pt);
\draw[line width=0.3mm] (-1,-2) circle (10pt);
\draw[line width=0.3mm] (1.5,0) circle (10pt);
\draw[line width=0.3mm] (1.5,-2) circle (10pt);
\draw[->,color = blue!50!black,line width=0.3mm] (0.5, 0.65)coordinate (a_1) --(1.4,-1.65) coordinate (a_2);
\end{tikzpicture} }
\caption{Graph of the model $\mathcal{N}_{CII}$.} \label{small2}
\end{figure}
\end{Example}
{Using} the em-algorithm described in Section \ref{SectEm} we took 500 random elements of $\mathcal{N}_{CIS}$ and calculated the closest element of $\mathcal{N}_{CII}$ by using the minimum KL-divergence of 50 different random input distributions in each run. The results are displayed in Table \ref{table}.
\begin{table}[H]
\caption{The results of the em-algorithm between $\mathcal{N}_{CIS}$ and $\mathcal{N}_{CII}$.} \label{table}
\centering
\begin{tabular}{clll}
\textbf{$\vert \mathcal{W} \vert $ }	& Minimum	& Maximum & Arithmetic Mean\\
\hline
 2 & 0.011969035529826939 & 0.5028091152589176 &  0.15263592877594967 \\
 3 & 0.021348311360946  & 0.5499395859771526 & 0.1538653506807848  \\
 4 & 0.014762084688030863 & 0.3984635189946462 & 0.15139198568055212 
 \\
 8 & 0.017334311629729246 & 0.4383731978333986 & 0.15481967618112732\\
 16 & 0.024306996171092318 & 0.4238222051787452 & 0.1490336847067273 \\
300 & 0.016524177216064712 &  0.47733473380366764 & 0.15493896625208842 \\
\end{tabular}
\end{table}
{This is an example of an element lying in $\mathcal{M}_{CIS}$, which cannot be approximated by an element in $\mathcal{M}_{CII}$.

Now we are going to look at this example from the causal perspective. Proposition \ref{IfandOnlyIf} states that $\Phi_{CII}(\tilde{P})$ is $0$ if and only if $\tilde{P}$ is the limit of a sequence of distributions in $\mathcal{M}_{CII}$ corresponding to distributions on the extended space that factor according to the split model. Hence a distribution resulting in $\Phi_{CII} > 0$ cannot be explained by a split model with an exterior influence. Taking into account that $\mathcal{M}_{CIS}$ does not correspond to a graph, we do not have a similar result describing the distributions for which $\Phi_{CIS} = 0$. Nonetheless, by looking at the graphical model $\mathcal{N}_{CIS}$, we are able to discuss the causal structure of a submodel of $\mathcal{M}_{CIS}$, a class of distributions for which $\Phi_{CIS} = 0$ holds.}

If we trust the results in Table \ref{table}, this would imply that the influence from $Y_{2}$ to $Y_{1}$ is not purely external, but that there suddenly develops an internal influence in timestep $t+1$ that did not exist in timestep $t$. {Therefore the distributions in $\mathcal{N}_{CIS}$ do not belong to the stationary Markovian processes $MP(\mathcal{Z})$, depicted in Figure \ref{interactInt}, in general. For these Markovian processes the connections between the $Y_{i}$s arise from correlated $X_{i}$s or external influences, as pointed out by Amari in Section 6.9 \cite{amari}. So from a causal perspective $\mathcal{N}_{CIS}$ does not fit into our framework. Hence the initial distribution $\tilde{P}$, which corresponds to a full model, will in general not be an element of $\mathcal{N}_{CIS}$. However, the projection of $\tilde{P}$ to $\mathcal{M}_{CIS}$ might lie in $\mathcal{N}_{CIS}$ as illustrated in Figure \ref{ncii}. 
\begin{figure}[H]
\centering
\begin{tikzpicture}
\draw[lightgray, fill= lightgray] plot [smooth cycle] coordinates
  {(0,1.5) (1,1.5) (4,1.5) (6,1.5) (5,-1) (3, 0) (0.5, 0.5 )};
\draw[fill=lightgray, lightgray] (-0.15,2) rectangle (6.01,1.5);
\draw[lightgray!20!white, fill= lightgray!20!white, opacity = 0.9] plot [smooth cycle] coordinates
  {(0,1.5) (2,1.5) (2.5,0.5) (3.1,-1) (2,-1) (1,-0.75) (0,-0.5)};
\draw[lightgray, fill= lightgray] (-0.15,1.65) rectangle (0.49,0.7);
\draw[lightgray!20!white,fill=  lightgray!20!white,  opacity=0.9] (-0.15,1.65) rectangle (0.5,0.3);
\draw[color=black] (4.5,1.5) node {$MP(\mathcal{Z})$};
\draw[color=black] (4,0) node {$\tilde{P}$};
\draw[color=black] (3.75,-0.05) node {$\bullet$};
\draw[lightgray!60!white, fill= lightgray!60!white,  opacity=0.8] plot [smooth cycle] coordinates
  {(1.75,0.25) (2.5,0.5) (3,-0.75) (2.5, -0.85) (1.5, -0.65)};
  \draw[color=black] (0.5,0.65) node {$\mathcal{M}_{CIS}$};
    \draw[color=black] (2,-0.5) node {$\mathcal{N}_{CIS}$};
\draw[color=black] (2.82,-0.25) node {$\bullet$};
\draw[line width = 0.3mm, black, ->] (3.75,-0.05) .. controls (3.5,-0.05) ..(2.95,-0.2);
\end{tikzpicture}
\caption{Sketch of the relationships among $MP(\mathcal{Z}), \mathcal{M}_{CIS}$ and $\mathcal{N}_{CIS}.$} \label{ncii}
\end{figure}
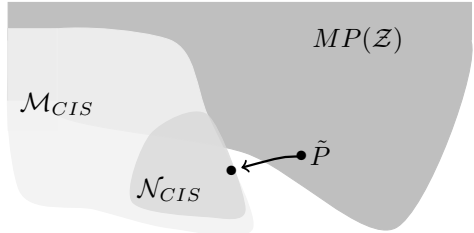

When this is the case, then $\tilde{P}$ is closer to an element with a causal structure that does not fit into the discussed setting, than to a split model in which only the causal cross-connections are removed. Hence a part of the internal cross-connections is being compensated by this type of model and therefore this does not measure all the intrinsic integrated information. }

Further examples, which hint towards $\mathcal{M}_{CII} \subsetneq \mathcal{M}_{CIS}$, can be found in Section \ref{sectResults}. 

Adding the hidden variable $W$ seems not to be sufficient to approximate elements of $\mathcal{M}_{CIS}$. Now the question naturally arises whether there are other exterior influences that need to be included in order to be able to approximate $\mathcal{M}_{CIS}$. We will explore this thought by starting with the graph corresponding to the split model $\mathcal{M}_{SI}$, depicted in Figure \ref{startinggraph01} on the left. In the next step we add hidden vertices and edges to the graph in a way such that the whole graph is still a chain graph. An example for a valid hidden structure is given in Figure \ref{startinggraph01} in the middle. Since we are going to marginalize over the hidden structure, it is only important how the visible nodes are connected via the hidden nodes. In the case of the example in Figure \ref{startinggraph01} we have a directed path from $X_{1}$ to $X_{2}$ going through the hidden nodes. Therefore we are able to reduce the structure to a gray box shown on the right in Figure \ref{startinggraph01}. 
\begin{figure}[H]
\centering
\begin{tikzpicture}
\draw[line width=0.5mm] (-7.5,0) node {\footnotesize $X_{1}$};
\draw[line width=0.5mm] (-7.5,-1) node {\footnotesize $X_{2}$};
\draw[line width=0.5mm] (-5.4, 0) node {\footnotesize $Y_{1}$};
\draw[line width=0.5mm] (-5.4, -1) node {\footnotesize $Y_{2}$};
\draw[line width = 0.45mm] (-7.5,-0.3)--(-7.5,-0.7);
\draw[->,line width=0.45mm] (-7.2,0)--(-5.75,0);
\draw[->,,line width=0.45mm] (-7.2,-1)--(-5.75,-1);
\draw[line width=0.3mm] (-7.5,0) circle (8pt);
\draw[line width=0.3mm] (-7.5,-1) circle (8pt);
\draw[line width=0.3mm] (-5.4,0) circle (8pt);
\draw[line width=0.3mm] (-5.4,-1) circle (8pt);

\node[single arrow,draw=black,fill=black!10,minimum height=1cm,shape border rotate=0] at (-4.25,-0.5) {};
\node[single arrow,draw=black,fill=black!10,minimum height=1cm,shape border rotate=0] at (0.25,-0.5) {};

\draw[ gray!60, fill = gray!20]  (-3.05,2) rectangle (-0.95,0.4);
\draw[line width=0.5mm] (-3,0) node {\footnotesize $X_{1}$};
\draw[line width=0.5mm] (-3,-1) node {\footnotesize $X_{2}$};
\draw[line width=0.5mm] (-0.9, 0) node {\footnotesize $Y_{1}$};
\draw[line width=0.5mm] (-0.9, -1) node {\footnotesize $Y_{2}$};
\draw[->,line width=0.3mm] (-2.9 ,0.275)--(-2.75,0.5);
\draw[->,line width=0.3mm] (-2.4,1.03)--(-2.25,1.25);
\draw[->,line width=0.3mm] (-1.3, 0.55)--(-1.1,0.3);
\draw[->,line width=0.3mm] (-2.235, 0.75)--(-1.85,0.75);
\draw[line width=0.3mm] (-1.8, 1.279)--(-1.6,1);
\draw[line width=0.5mm] (-2.5, 0.75) node {\footnotesize $W_{1}$};
\draw[line width=0.5mm] (-1.5, 0.75) node {\footnotesize $W_{2}$};
\draw[line width=0.5mm] (-2, 1.5) node {\footnotesize $W_{3}$};
\draw[line width=0.3mm] (-2.5, 0.75) circle (8pt);
\draw[line width=0.3mm] (-1.5, 0.75) circle (8pt);
\draw[line width=0.3mm] (-2,1.5) circle (8pt);
\draw[line width = 0.45mm] (-3,-0.3)--(-3,-0.7);
\draw[->,line width=0.45mm] (-2.7,0)--(-1.25,0);
\draw[->,,line width=0.45mm] (-2.7,-1)--(-1.25,-1);
\draw[line width=0.3mm] (-3,0) circle (8pt);
\draw[line width=0.3mm] (-3,-1) circle (8pt);
\draw[line width=0.3mm] (-0.9,0) circle (8pt);
\draw[line width=0.3mm] (-0.9,-1) circle (8pt);

\draw[line width=0.5mm] (1.5,0) node {\footnotesize $X_{1}$};
\draw[line width=0.5mm] (1.5,-1) node {\footnotesize $X_{2}$};
\draw[line width=0.5mm] (3.6, 0) node {\footnotesize $Y_{1}$};
\draw[line width=0.5mm] (3.6, -1) node {\footnotesize $Y_{2}$};
\filldraw [gray!80, fill = gray!40] (2.3,0.6) rectangle (2.8,0.4);
\draw[line width = 0.45mm] (1.5,-0.3)--(1.5,-0.7);
\draw[<-,,line width=0.25mm] (3.3,0.15)--(2.8,0.4);
\draw[->,,line width=0.25mm] (1.75,0.15)--(2.25,0.4);
\draw[fill = white, white] (2.55,-0.5) circle (2pt);
\draw[->,,line width=0.45mm] (1.8,0)--(3.25,0);
\draw[->,,line width=0.45mm] (1.8,-1)--(3.25,-1);
\draw[line width=0.3mm] (1.5,0) circle (8pt);
\draw[line width=0.3mm] (1.5,-1) circle (8pt);
\draw[line width=0.3mm] (3.6,0) circle (8pt);
\draw[line width=0.3mm] (3.6,-1) circle (8pt);
\end{tikzpicture}
\caption{Example of an exterior influence on the initial graph.} \label{startinggraph01}
\end{figure}
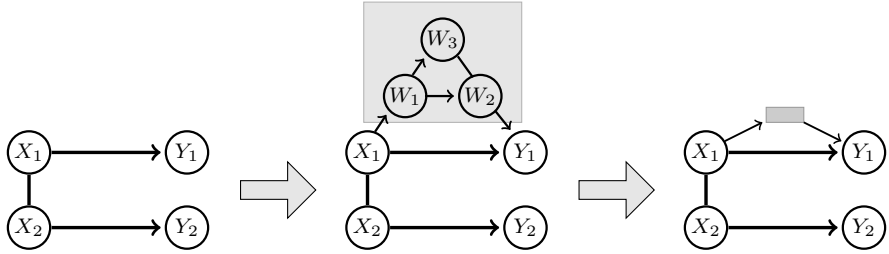
Then we use the Algorithm \ref{Algorithm} mentioned earlier, which converts a chain graph with hidden variables to a chain mixed graph reflecting the conditional independence structure of the marginalized model. This leads to a directed edge from $X_{1}$ to $X_{2}$  by marginalizing over the nodes in the hidden structures.
Seeing that this directed edge already existed, the resulting model now is a subset of $\mathcal{M}_{SI}$ and therefore does not approximate $\mathcal{M}_{CIS}$. 

Following this procedure we are able to show that adding further hidden nodes and subgraphs of hidden nodes does not lead to a chain mixed graph belonging to a model that satisfies the conditional independence statements of Property \ref{postulate1} and strictly contains $\mathcal{M}_{CII}$. 
\begin{Theorem} \label{proposition}
It is not possible to create a chain mixed graph corresponding to a model $\mathcal{M}$, such that its distributions satisfy Property \ref{postulate1} and $\mathcal{M}_{CII} \subsetneq \mathcal{M}$, by introducing a more complicated hidden structure to the graph of $\mathcal{M}_{SI}$.
\end{Theorem}

In conclusion, assuming that Conjecture \ref{Conj} holds, we have the following relations among the different presented models.
\begin{equation*}
\begin{split}
&\mathcal{M}_{I} \subsetneq \mathcal{M}_{G} \\
&\mathcal{M}_{I} \subsetneq \mathcal{M}_{CII} \subsetneq \mathcal{M}_{CIS} \\
&\mathcal{M}_{SI} \subsetneq \mathcal{M}_{CII} \subsetneq \mathcal{M}_{CIS}
\end{split}
\end{equation*}
A sketch of the inclusion properties among the models is displayed in Figure \ref{relationship}.
\begin{figure}[H]
\centering
\scalebox{0.9}{
\begin{tikzpicture}
    \draw[fill = gray, opacity = 0.2] (-1,0) ellipse (3cm and 2cm);
    \draw[fill = gray, opacity = 0.3] (-1.5,0) ellipse (2.25cm and 1.5cm);
    \draw[] (-0.25,0) node{$\mathcal{M}_{CII}$};
     \draw[] (1.4,0) node{$\mathcal{M}_{CIS}$};
	
     \draw[fill = black, opacity = 0.2] (-2,1.5) ellipse (1cm and 2cm);
     \draw[gray, fill = gray, opacity = 0.6] (-2,0.45) ellipse (0.8cm and 0.75cm);
     \draw[fill = gray, gray, opacity = 0.5] (-2,-0.35) ellipse (1cm and 0.8cm);
     \draw[] (-2,0.75) node{$\mathcal{M}_{I}$};
     \draw[] (-2,3) node{$\mathcal{M}_{G}$};
     \draw[] (-2,-0.8) node{$\mathcal{M}_{SI}$};
\end{tikzpicture} }
\caption{Sketch of the relationship between the manifolds corresponding to the different measures. } \label{relationship}
\end{figure}
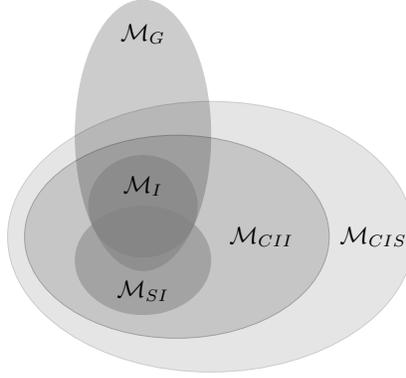
Every set that lies inside $\mathcal{M}_{CIS}$ satisfies Property \ref{postulate1} and every set that completely contains $\mathcal{M}_{I}$ fulfills Property \ref{postulate2}.

\subsubsection{em-Algorithm} \label{SectEm}
The calculation of the measure $\Phi_{CII}^{m}$ with
\begin{equation*}
\Phi_{CII}^{m} = \inf\limits_{Q \in \mathcal{M}_{CII}^{m}} D_{\mathcal{Z}}(\tilde{P} \parallel Q)
\end{equation*} 
can be done by the em-algorithm, a well known information geometric algorithm. It was proposed by Csiszár and Tusnády in 1984 in Reference \cite{Csiszar} and its usage in the context of neural networks with hidden variables was described for example by Amari et al.~in Reference \cite{EM}. The expectation-maximization EM-algorithm \cite{Dempster} used in statistics is equivalent to the em-algorithm in many cases, including this one, as we will see below. A detailed discussion of the relationship of these algorithms can be found in Reference \cite{Emem}.

In order to calculate the distance between the distribution $\tilde{P}$ and the set $\mathcal{M}_{CII}^{m}$ on $\mathcal{Z}$ we will make use of the extended space of distributions on $\mathcal{Z} \times \mathcal{W}^{m}$, $\mathcal{P}(\mathcal{Z} \times \mathcal{W}^{m})$.
Let $\mathcal{M}_{W|Z}$ be the set of all distributions on $\mathcal{Z} \times \mathcal{W}^{m}$ that have $\mathcal{Z}$-marginals equal to the distribution of the whole system $\tilde{P}$
\begin{equation*}
\begin{split}
\mathcal{M}_{W|Z} =& \left\lbrace P \in \mathcal{P}(\mathcal{Z} \times \mathcal{W}^{m}) \mid P(z) = \tilde{P}(z), \, \forall z \in \mathcal{Z} \right\rbrace \\
 =& \left\lbrace P \in \mathcal{P}(\mathcal{Z} \times \mathcal{W}^{m}) \mid P(z,w) = \tilde{P}(z)P(w \vert z), \, \forall (z,w) \in \mathcal{Z} \times \mathcal{W}^{m} \right\rbrace.
 \end{split}
\end{equation*} 
This is an $ m$-flat submanifold since it is linear w.r.t $P(w \vert z)$. Therefore there exists a unique $e$-projection to $\mathcal{M}_{W|Z}$.

The second set that we are going to use is the set $\mathcal{E}^{m}$ of distributions that factor according to the split model including the common exterior influence. We have seen this set before in Section \ref{Sectgroundtruth}.
\begin{equation}
\begin{split}
\mathcal{E}^{m} &= \left\lbrace P \in \mathcal{P}(\mathcal{Z} \times \mathcal{W}^{m}) \mid P(z,w) = P(x) \prod\limits_{i=1}^{n}P(y_{i} \vert x_{i}, w)P(w), \, \forall (z,w) \in \mathcal{Z} \times \mathcal{W}^{m} \right\rbrace  \label{setE}.
\end{split}
\end{equation}
This set is in general not $e$-flat, but we will show that there is a unique $m$-projection to it.
We are able to use these sets instead of $\tilde{P}$ and $\mathcal{M}_{CII}^{m}$ because of the following result.

\begin{Theorem}[Theorem 7 from Reference \cite{EM}] \label{ThmEquivMin}
The minimum divergence between $\mathcal{M}_{W|Z}$ and $\mathcal{E}^{m}$ is equal to the minimum divergence between $\tilde{P}$ and $\mathcal{M}_{CII}^{m}$ in the visible manifold
\begin{equation*}
\inf\limits_{P \in \mathcal{M}_{W|Z}, Q \in \mathcal{E}^{m}} D_{\mathcal{Z}\times \mathcal{W}^{m}}(P \parallel Q) = \inf\limits_{\tilde{Q} \in \mathcal{M}_{CII}^{m}} D_{\mathcal{Z}}(\tilde{P} \parallel \tilde{Q}).
\end{equation*}
\end{Theorem} \label{equivalence}
\begin{proof}[Proof of Theorem \ref{ThmEquivMin}]
Let $P,Q \in \mathcal{P}(\mathcal{Z} \times \mathcal{W}^{m})$, using the chain-rule for KL-divergence leads to
\begin{equation*}
D_{\mathcal{Z}\times \mathcal{W}^{m}}(P \parallel Q) = D_{\mathcal{Z}}(P \parallel Q) + D_{\mathcal{W} \vert \mathcal{Z}}(P \parallel Q),
\end{equation*}
with 
\begin{equation*}
 D_{\mathcal{W} \vert \mathcal{Z}}(P \parallel Q) = \sum\limits_{(z, w) \in \mathcal{Z} \times\mathcal{W}^{m}} P(z,w) log \, \dfrac{P(w \vert z)}{Q(w \vert z)}.
\end{equation*}
This results in 
\begin{equation*}
\begin{split}
\inf\limits_{P \in \mathcal{M}_{W|Z}, Q \in \mathcal{E}^{m}} D_{\mathcal{Z} \times \mathcal{W}^{m}}(P \parallel Q) &= \inf\limits_{P \in \mathcal{M}_{W|Z} ,Q \in \mathcal{E}^{m}} \left\lbrace  D_{\mathcal{Z}}(P \parallel Q) + D_{\mathcal{W} \vert \mathcal{Z}}(P \parallel Q) \right\rbrace \\
&= \inf\limits_{P \in \mathcal{M}_{W|Z} ,Q \in \mathcal{E}^{m}} \left\lbrace  D_{\mathcal{Z}}(\tilde{P} \parallel Q) + D_{\mathcal{W} \vert \mathcal{Z}}(P \parallel Q) \right\rbrace \\
& =\inf\limits_{Q \in \mathcal{M}_{CII}^{m}} D_{\mathcal{Z}}(\tilde{P} \parallel Q). 
\end{split}
\end{equation*}
\end{proof}

The em-algorithm is an iterative algorithm that first performs an $e$-projection to $\mathcal{M}_{W|Z}$ and then an $m$-projection to $\mathcal{E}^{m}$ repeatedly. Let $Q_{0} \in \mathcal{E}^{m}$ be an arbitrary starting point and define $P_{1}$ as the $e$-projection of $Q_{0}$ to $\mathcal{M}_{W|Z}$
\begin{equation*}
P_{1} = \arginf\limits_{P \in \mathcal{M}_{W|Z}} D_{\mathcal{Z} \times \mathcal{W}^{m}}(P \parallel Q_{0}).
\end{equation*}
Now we define $Q_{1}$ as the $m$-projection of $P_{1}$ to $\mathcal{E}^{m}$
\begin{equation*}
Q_{1} = \arginf\limits_{Q \in \mathcal{E}^{m}} D_{\mathcal{Z} \times \mathcal{W}^{m}}(P_{1} \parallel Q).
\end{equation*}
Repeating this leads to
\begin{equation*}
P_{i+1} = \arginf\limits_{P \in \mathcal{M}_{W|Z}} D_{\mathcal{Z} \times \mathcal{W}^{m}}(P \parallel Q_{i}), \quad
Q_{i+1} = \arginf\limits_{Q \in \mathcal{E}^{m}} D_{\mathcal{Z} \times \mathcal{W}^{m}}(P_{i+1} \parallel Q).
\end{equation*}
The correspondence between these projections in the extended space $\mathcal{P}(\mathcal{Z} \times \mathcal{W}^{m})$ and one $m$-projection in $\mathcal{P}(\mathcal{Z})$ is illustrated in Figure  \ref{emprojj}.
\begin{figure}[H]
\centering
\includegraphics[scale=0.3]{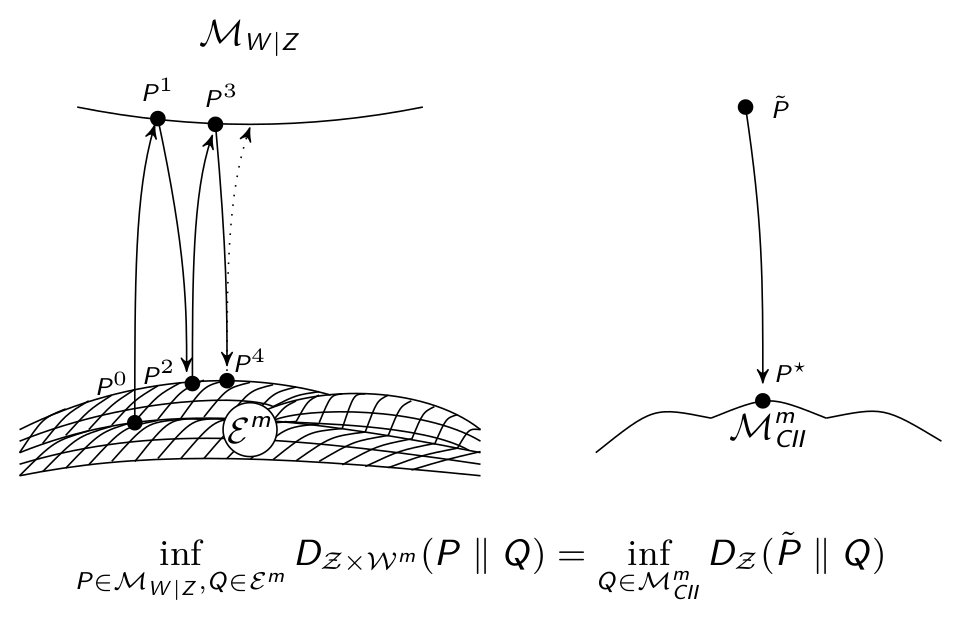}
\caption{Sketch of the em-Algorithm.}\label{emprojj}
\end{figure}
The algorithm iterates between the extended spaces $\mathcal{M}_{W \vert Z}$ and $\mathcal{E}^{m}$ on the left of Figure \ref{emprojj}. Using Theorem \ref{equivalence} we gain that this minimization is equivalent to the minimization between $\tilde{P}$ and $\mathcal{M}_{CII}^{m}$.
The convergence of this algorithm is given by the following result.
\begin{Proposition}[Theorem 8 from Reference \cite{EM}] \label{emProp}
The monotonic relations 
\begin{equation*}
D_{\mathcal{Z} \times \mathcal{W}^{m}}(P_{i} \parallel Q_{i}) \geq D_{\mathcal{Z} \times \mathcal{W}^{m}}(P_{i+1} \parallel Q_{i}) \geq D_{\mathcal{Z} \times \mathcal{W}^{m}}(P_{i+1} \parallel Q_{i+1})
\end{equation*}
hold, where equality holds only for the fixed points $(\hat{P}, \hat{Q}) \in \mathcal{M}_{W|Z} \times \mathcal{E}^{m}$ of the projections
\begin{equation*}
\begin{split}
\hat{P} &= \arginf\limits_{P \in \mathcal{M}_{W|Z}} D_{\mathcal{Z} \times \mathcal{W}^{m}}(P \parallel \hat{Q}) \\
\hat{Q} &= \arginf\limits_{Q \in \mathcal{E}^{m}} D_{\mathcal{Z} \times \mathcal{W}^{m}}(\hat{P} \parallel Q).
\end{split}
\end{equation*}
\end{Proposition}
\begin{proof}[Proof of Proposition \ref{emProp}]
This is immediate, because of the definitions of the $e$- and $m$-projections.
\end{proof}
Hence this algorithm is guaranteed to converge towards a minimum, but this minimum might be local. We will see examples of that in Section \ref{sectResults}.

In order to use this algorithm to calculate $\Phi_{CII}$ we first need to determine how to perform an $e$- and $m$-projection in this case.
The $e$-projection from $Q \in \mathcal{E}^{m}$ to $\mathcal{M}_{W|Z}$ is given by
\begin{equation*}
P(z,w) = \tilde{P}(z) Q(w \vert z), 
\end{equation*}
for all $(z,w) \in \mathcal{Z} \times \mathcal{W}^{m}$. This is the projection because of the following equality
\begin{equation*}
\begin{split}
D_{\mathcal{Z} \times \mathcal{W}^{m}}(P \parallel Q) &= \sum\limits_{(z,w) \in \mathcal{Z} \times \mathcal{W}^{m}} P(z,w) \, log \, \dfrac{P(z,w)}{Q(z,w)} \\
&= \sum\limits_{z \in \mathcal{Z}} \tilde{P}(z) \, log \, \dfrac{\tilde{P}(z)}{Q(z)} + \sum\limits_{(z,w) \in \mathcal{Z} \times \mathcal{W}^{m}} P(z,w) \, log \, \dfrac{P(w \vert z)}{Q(w \vert z)}.
\end{split}
\end{equation*}
The first addend is a constant for a fixed distribution $\tilde{P}$ and the second addend is equal to 0 if and only if $P(w \vert z) = Q(w \vert z)$. Note that this means that the conditional expectation of $W$ remains fixed during the $e$-projection. This is an important point, because this guarantees the equivalence to the EM algorithm and therefore the convergence towards the MLE. For a proof and examples see Theorem 8.1 in Reference \cite{amari} and Section 6 in Reference \cite{Emem}.

After discussing the $e$-projection, we now consider the $m$-projection.
\begin{Proposition} \label{ColmProj}
The $m$-projection from $P \in \mathcal{M}_{W|Z}$ is given by
\begin{equation*}
Q(z,w) = P(x) \prod\limits_{i=1}^{n}P(y_{i} \vert x_{i}, w)P(w)
\end{equation*}
for all $(z,w) \in \mathcal{Z} \times \mathcal{W}^{m}$. 
\end{Proposition}

The last remaining decision to be made before calculating $\Phi_{CII}$ is the choice of the initial distribution. Since it depends on the initial distribution whether the algorithm converges towards a local or global minimum, it is important to take the minimal outcome of multiple runs.
One class of starting points that immediately lead to an equilibrium, which is in general not minimal, are the ones in which $Z$ and $W$ are independent $P^{0}(z,w) = P^{0}(z) P^{0}(w)$. It is easy to check that the algorithm converges here to the fixed point $\hat{P}$
\begin{equation*}
\begin{split}
\hat{P}(z,w) &=  \tilde{P}(x) \dfrac{1}{\vert \mathcal{W}^{m} \vert} \prod\limits_{i}^{n} \tilde{P}(y_{i} \vert x_{i}) \\
\hat{P}(z) &= \tilde{P}(x) \prod\limits_{i}^{n} \tilde{P}(y_{i} \vert x_{i}).
\end{split}
\end{equation*}
Note that this is the result of the $m$-projection of $\tilde{P}$ to $\mathcal{M}_{SI}$, the manifold belonging to $\Phi_{SI}$.

\subsection{Comparison} \label{SectComp}
In order to compare the different measures, we need a setting in which we generate the probability distributions of full systems. We chose to use weighted Ising models as described in the next section.
\subsubsection{Ising Model}
The distributions used to compare the different measures in the next chapter are generated by weighted Ising models, 
also known as binary auto-logistic models as described in Reference \cite{Boltz} Example 3.2.3. Let us consider $n$ binary variables $X = (X_{1}, \dots, X_{n})$, $\mathcal{X} = \{-1,1\}^{n}$. The matrix $V \in \mathbb{R}^{n \times n}$ contains the weights $v_{i j}$ of the connection from $X_{i}$ to $Y_{j}$ as displayed in Figure \ref{Boltzmann}. Note that this figure is not a graphical model corresponding to the stationary distribution, but merely displays the connections of the conditional distribution of $Y_{i} = y_{i}$ given $X = x$ with the respective weights 
\begin{equation}
P(y_{j} \vert x ) = \dfrac{1}{ 1 + e^{- 2 \beta \sum\limits_{i=1}^{n} v_{i j} x_{i} y_{j} } }. \label{kernel}
\end{equation}
The inverse temperature $\beta > 0$ regulates the coupling strength between the nodes. For $\beta $ close to zero the different nodes are almost independent and as $\beta$ grows the connections become stronger.

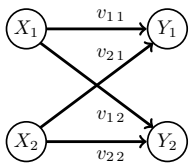
\begin{figure}[H]
\centering
\scalebox{0.75}{
\begin{tikzpicture}[rounded corners]
\draw[] (0,0) node {$X_{1}$};
\draw[] (1.5,0.25) node {$v_{1 \, 1}$};
\draw[] (1.5,-0.45) node {$v_{2 \, 1}$};
\draw[] (1.5,-2.25) node {$v_{2 \, 2}$};
\draw[] (1.5,-1.55) node {$v_{1 \, 2}$};
\draw[->,line width=0.5mm] (0.35,0)--(2.133,0);
\draw[->,line width=0.5mm] (0.35,-2)--(2.133,-2);
\draw[->, line width=0.5mm] (0.25,-0.25)--(2.25,-1.75);
\draw[->,line width=0.5mm] (0.25,-1.75)--(2.25,-0.25);
\draw[] (0,-2) node {$X_{2}$};
\draw[] (2.5,0) node { $Y_{1}$};
\draw[] (2.5,-2) node { $Y_{2}$};
\draw[line width=0.3mm] (0,0) circle (10pt);
\draw[line width=0.3mm] (0,-2) circle (10pt);
\draw[line width=0.3mm] (2.5,0) circle (10pt);
\draw[line width=0.3mm] (2.5,-2) circle (10pt);

\end{tikzpicture} }
\caption{The weights corresponding to the connections for $n = 2$. } \label{Boltzmann}
\end{figure}
We are calculating the stationary distribution $\hat{P}$ by starting with a random initial distribution $P^{0}$ and then multiplying by (\ref{kernel}) in the following way
\begin{equation*}
P^{t+1}(x) = \sum\limits_{x \in \mathcal{X}} P^{t}(x) \cdot \prod\limits_{j = 1}^{n} P(y_{i} \vert x),
\end{equation*}
this leads to
\begin{equation*}
\hat{P} = \lim\limits_{t \rightarrow \infty} P^{t}.
\end{equation*}
There always exists a unique stationary distribution, see for instance Reference \cite{Boltz}, Theorem 5.1.2 . 

\subsubsection{Results} \label{sectResults}

In this section we are going to compare the different measures experimentally. {Note that we do not have an exterior influence in these examples, so that $\Phi_{T} = \Phi_{SI}$ holds.}

To distinguish between the Causal Information Integration $\Phi_{CII}$ calculated with different sized state spaces of $W$, we will denote 
\begin{equation*}
\Phi_{CII}^{m} = \inf\limits_{Q \in \mathcal{M}_{CII}^{m}} D_{\mathcal{Z}}(\tilde{P} \parallel Q).
\end{equation*}  
We start with the smallest example possible, with $n = 2$, and the weight matrix
\begin{equation*}
V = \begin{pmatrix}
0.0084181 & -0.2401545 \\
0.39270161 & 0.37198751
\end{pmatrix}
\end{equation*}
shown in Figure \ref{Boltz2}. In this example every measure is bounded by $\Phi_{I}$ and the measures $\Phi_{I}, \Phi_{G}$ and $\Phi_{SI}$ display a limit behavior different from $\Phi_{CIS}$ and the $\Phi_{CII}$. The state spaces of $W$ have the size 2, 3, 4, 36 and 92 and the respective measures are displayed in shades of blue that get darker as the state space gets larger. In every case the em-algorithm has been initiated 100 times with a random input distribution in order to find a global minimum. {Minimizing over the outcome of 100 different runs turns out to be sufficient, at least empirically, to reveal the behavior of the global minima.} On the right side of this figure, we are able to see the difference between $\Phi_{CIS}$ and $\Phi_{CII}$. Considering the precision of the algorithms we assume that a difference smaller than 5e-07 is approx.~zero. We can see that in a region from $\beta = 15$ to $\beta = 25$ the measures differ even in the case of 92 hidden states. So this small case already hints towards $\mathcal{M}_{CII} \subsetneq \mathcal{M}_{CIS}$.
\begin{figure}[H]
\includegraphics[width = 0.9\textwidth]{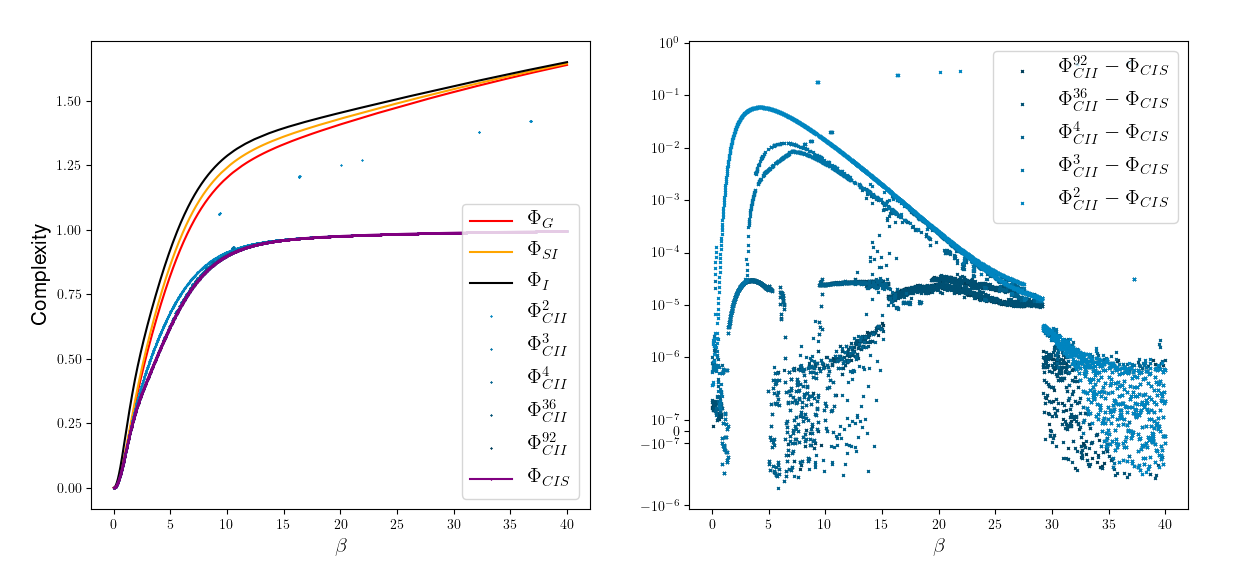}
\caption{Ising model with 2 nodes and the differences between $\Phi_{CIS}$ and $\Phi_{CII}$} \label{Boltz2}
\end{figure}
Increasing $n$ from 2 to 3 makes the difference even more visible, as we can see in Figure \ref{Ex3} produced with the weight matrix 
\begin{equation*}
V = \begin{pmatrix}
-0.43478388 &  0.47448218 &  0.36808313 \\
 0.52117467 &  0.00672578 & -0.7387737  \\
 -0.56114795 & -0.96941243 & -0.76408711
\end{pmatrix}.
\end{equation*}

Here we are able to observe a difference in the behavior of $\Phi_{G}$ compared to the other measures, since we see that $\Phi_{I}, \Phi_{SI}$, $\Phi_{CII}$ and $\Phi_{G}$ are still increasing around $\beta \approx 1.1$, while $\Phi_{G}$ starts to decrease. 
\begin{figure}[H]
\centering
\includegraphics[width = 0.8\textwidth]{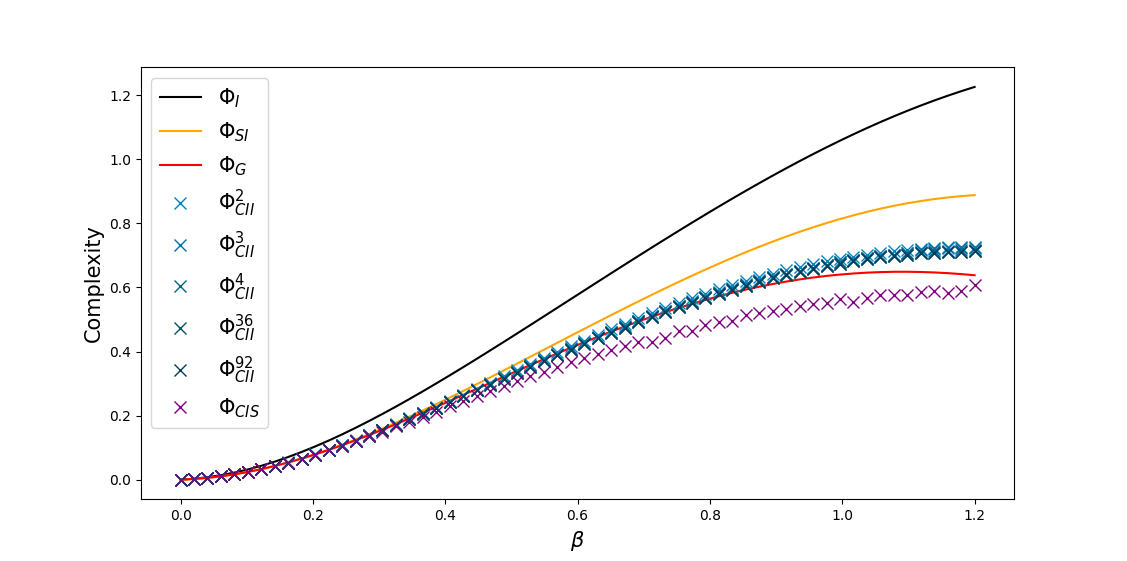}
\caption{Ising model with 3 nodes.} \label{Ex3}
\end{figure}

Now, we are going to focus on an example with 5 nodes. Since it is very time consuming to calculate $\Phi_{CIS}$ for more than 3 nodes, we are going to restrict attention to $\Phi_{I}$, $\Phi_{G}$, $\Phi_{SI}$ and $\Phi_{CII}$. 
The weight matrix
\begin{equation*}
V = \begin{pmatrix} 
-0.35615839 & -0.09775903 & 0.89743801 & -0.00604247 & -0.03897772 \\
-0.2260056 & 0.47769717 & -0.4302256 & 0.18692707 & 0.25140741 \\
-0.86081159 & -0.18348132 & -0.71528754 & -0.08100602 &  -0.64364176 \\
-0.13967234 & -0.03233011 & -0.81057654 & -0.33327558 & -0.57447322 \\
0.18920264 & -0.99054716 & 0.32088358 & 0.69100397 & -0.69206604
\end{pmatrix}
\end{equation*}
produces the Figure \ref{Ex5}. This example shows that $\Phi_{SI}$ is not bounded by $\Phi_{I}$ and therefore does not satisfy Property \ref{postulate2}. {Since the focus in this examples lies on the relationship between $\Phi_{SI}$ and $\Phi_{I}$, the em-algorithm was run with ten different input distributions for each step.}
\begin{figure}[H]
\centering
\includegraphics[width = 0.8\textwidth]{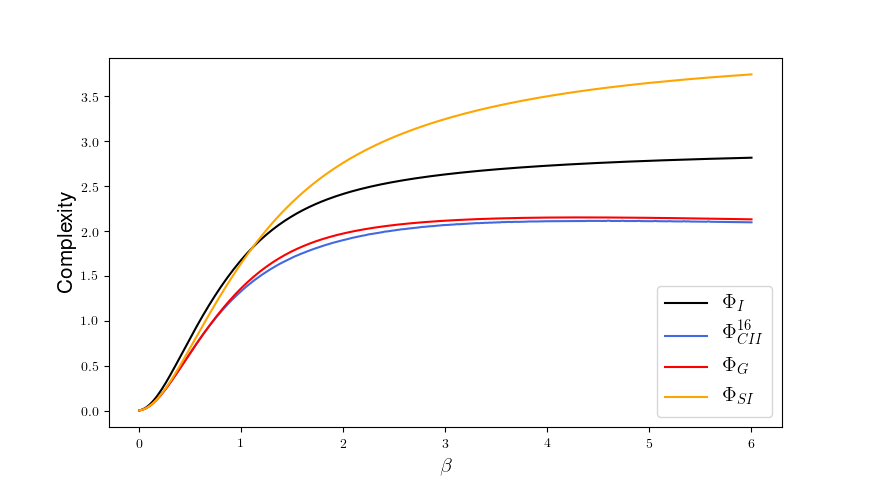}
\caption{Ising model with 5 nodes.} \label{Ex5}
\end{figure}
Using this example, we are going to take a closer look at the local minima the em-algorithm converges to. 
Considering only $\Phi_{CII}$ and varying the size of the state space leads to the upper part in Figure \ref{Boltz5diffState}. This figure displays ten different runs of the em-algorithm with each size of state space in different shades of the respective color, namely blue for $\Phi_{CII}^{2}$, violet for $\Phi^{4}_{CII}$, red for $\Phi^{8}_{CII}$ and orange for $\Phi^{16}_{CII}$. {Note that we display the outcomes of every run in this case and not only the minimal one, since we are interested in the local minima.} We are able to observe how increasing the state space leads to a smaller value of $\Phi_{CII}$. Additionally, the differences between the minimal values corresponding to each state space grow smaller and converge as the state spaces increase. 
\begin{figure}[H]
\includegraphics[width = \textwidth]{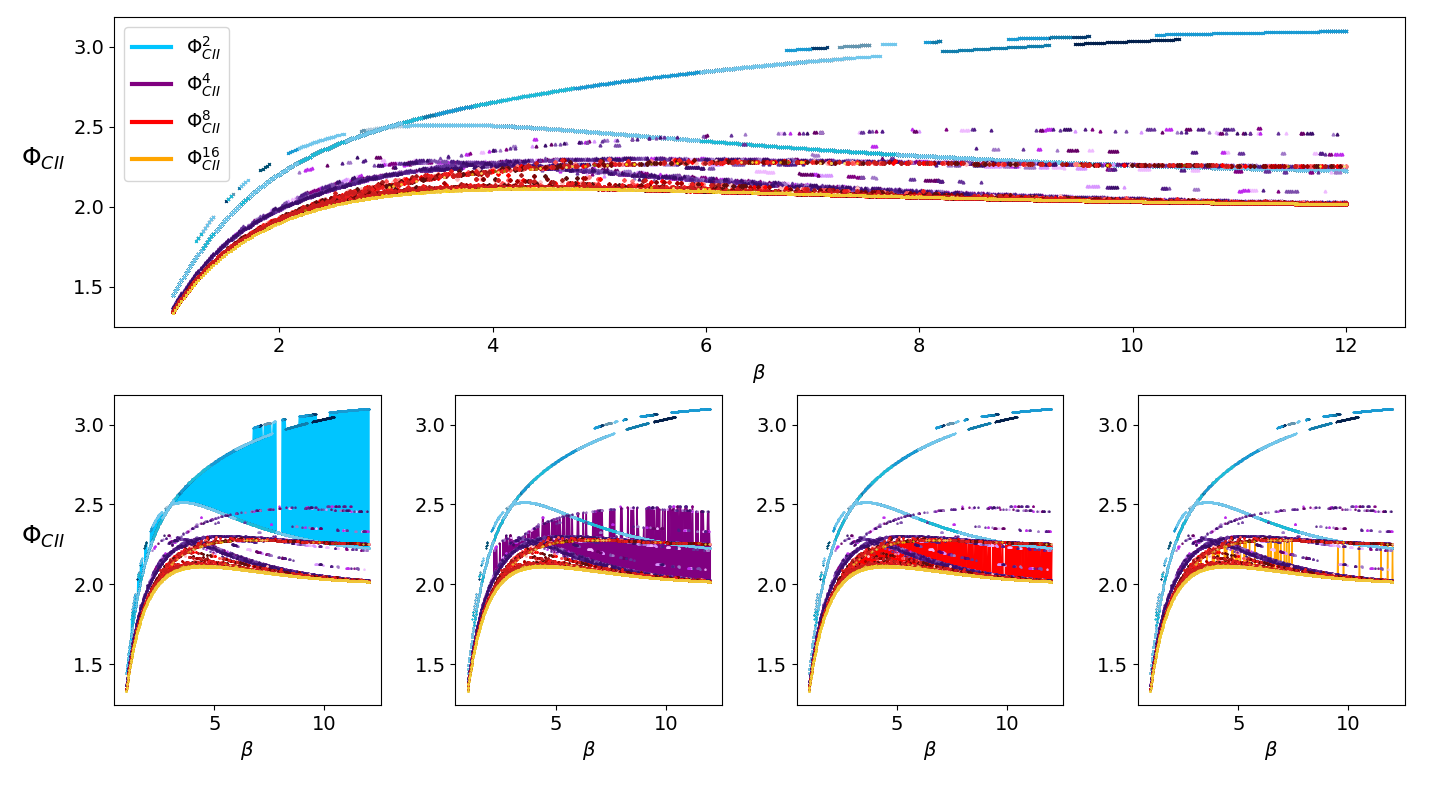}
\caption{The effect of a different sized state space.} \label{Boltz5diffState}
\end{figure}

The bottom half of Figure \ref{Boltz5diffState} highlights an observation that we made. Each of the four illustrations is a copy of the one above, where the difference between the minima are shaded in the respective color. By increasing the size of the state space the difference in value between the various local minima decreases visibly. We think this is consistent with the general observation made in the context of high dimensional optimization, for example,~Reference \cite{Multilayer} in which the authors conjecture that the probability of finding a high valued local minimum decreases when the network size grows. 

Letting the algorithm run only once with $\vert \mathcal{W}\vert = 2$ on the same data leads to a curve on the left in Figure \ref{DifflocMin}. 
\begin{figure}[H]
\includegraphics[width = \textwidth]{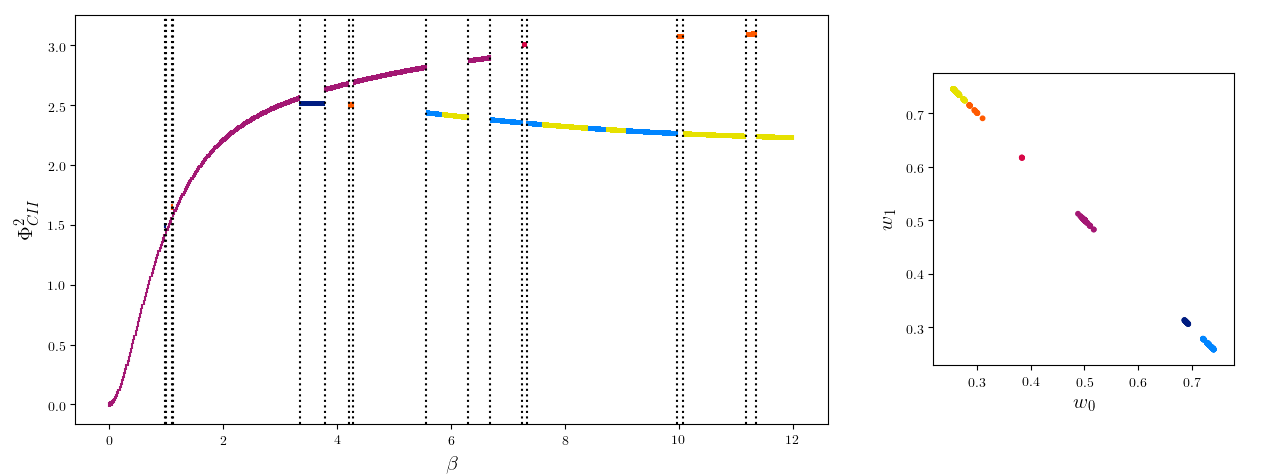}
\caption{Curve of one run of the em-algorithm for each $\beta$ coloured according to the distribution of $W$.} \label{DifflocMin}
\end{figure}
The sets $\mathcal{E}$ defined in (\ref{setE}) and $\mathcal{M}_{CII}$ (\ref{setMCII}) do not change for different values of $\beta$ and therefore we have a fixed set of local minima for a fixed state space of $W$. What does change with different $\beta$ is which of the local minima are global minima. The vertical dotted lines represent the steps $P^{\beta_{t}}$ to $P^{\beta_{t+1}}$ in which the KL-divergence between the projection to $\mathcal{M}_{CII}$ is greater than 0.2
\begin{equation*}
D_{\mathcal{Z}}(P^{\beta_{t}, \star} \parallel P^{\beta_{t+1}, \star}) > 0.2,
\end{equation*}
meaning that inside the different sections of the curve, the projections to $\mathcal{M}_{CII}$ are close. As $\beta$ increases, a different region of local minima becomes global. A sketch of this is shown in Figure \ref{SkDiffloc}.
\begin{figure}[H]
\centering
\includegraphics[scale=0.25]{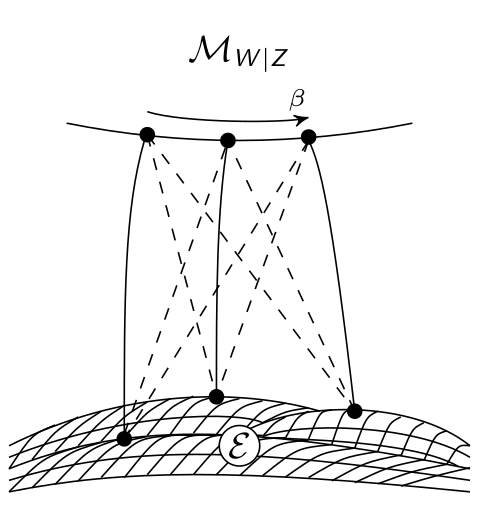}
\caption{Sketch of different local Minima.} \label{SkDiffloc}
\end{figure}

The curve is colored according to the distribution of $W$ as shown on the right side of Figure \ref{DifflocMin}. We see that a different distribution on $\mathcal{W}$ results in a different minimum, except for the region between ~7.5 and 8. The colors light blue and yellow refer to distributions on $\mathcal{W}$ that are different, but symmetric in the following way.  Consider two different distributions $Q, \hat{Q}$ on $\mathcal{Z} \times \mathcal{W}$ such that 
\begin{equation*}
Q(z,w_{1}) = \hat{Q}(z, w_{2}) \,  \text{  and  } \, Q(z,w_{2}) = \hat{Q}(z, w_{1})
\end{equation*} 
for all $z \in \mathcal{Z} $. Then the corresponding marginalized distributions in $\mathcal{M}_{CII}^{2}$ are equal
\begin{equation*}
\sum\limits_{w} Q(z,w) = \sum\limits_{w} \hat{Q}(z,w_{1}).
\end{equation*} 
This symmetry is the reason for the different colors in the region between ~7.5 and 8. 

Using this geometric algorithm we therefore gain a notion of the local minima on $\mathcal{E}$.

\section{Discussion}

This article discusses a selection of existing complexity measures in the context of Integrated Information Theory that follow the framework introduced in Reference \cite{stochInt}, namely $\Phi_{SI}, \Phi_{G}$ and $\Phi_{CIS}$. The main contribution is the proposal of a new measure, Causal Information Integration $\Phi_{CII}$. 

In Reference \cite{uniframe} and Reference \cite{GeomInfInt} the authors postulate a Markov condition, {ensuring the removal of the causal cross-connections}, and an upper bound, given by the mutual information $\Phi_{I}$, for valid Integrated Information measures. Although $\Phi_{SI}$ is not bounded by $\Phi_{I}$, as we see in Figure \ref{Ex5}, it does measure the intrinsic causal cross-connections in a setting in which there exists no common exterior influences. Therefore the authors of Reference \cite{Comparison} criticize this bound. Since wrongly assuming the existence of a common exterior influence might lead to a value that does not measure all the intrinsic causal influences, the question which measure to use strongly depends on how much we know about the system and its environment. We argue that using $\Phi_{I}$ as an upper bound in the cases in which we have an unknown common exterior influence is reasonable. The measure $\Phi_{G}$ attempts to extend $\Phi_{SI}$ to a setting with exterior influences, but it does not satisfy the Markov condition postulated in Reference \cite{uniframe}. 

One measure that fulfills all the requirements of this framework is $\Phi_{CIS}$, but it has no graphical representation. Hence the causal nature of the measured information flow is difficult to analyze. {We present in Example \ref{ex} a submodel of $\mathcal{M}_{CIS}$ that has a causal structure, which does not lie inside the set of Markovian processes $MP(\mathcal{Z})$, that we discuss in this article. Therefore by projecting to $\mathcal{M}_{CIS}$ we might project to a distribution that still holds some of the integrated information of the original system, although it does not have any causal cross-connections. Additionally we demonstrate that $\mathcal{M}_{CIS}$ does not correspond to a graphical representation, even after adding any number of latent variables to the model of $\mathcal{M}_{SI}$. This is conflicting with the strong connection between conditional independence statements and graphs in Pearls causality theory.} For discrete variables $\Phi_{CIS}$ does not have a closed form solution and has to be calculated numerically.

We propose a new measure $\Phi_{CII}$ that also satisfies all the conditions and has additionally a graphical and intuitive interpretation.  Numerically calculated examples indicate that $\Phi_{CII} \subsetneq \Phi_{CIS}$. The definition of $\Phi_{CII}$ explicitly includes an interior influence as a latent variable and therefore aims at only measuring intrinsic causal influences. This measure should be used in the setting in which there exists an unknown common exterior influence. By assuming the existence of a ground truth, we are able to prove that our new measure is bounded from above by the ultimate value of Integrated Information $\Phi_{T}$ of this system. Although $\Phi_{CII}$ also has no analytical solution, we are able to use the information geometric em-algorithm to calculate it. The em-algorithm is guaranteed to converge towards a minimum, but this might be local. {Even after letting our smallest example, depicted in Figure \ref{Boltz2}, run with 100 random input distributions, we still get local minima.} On the other hand, in our experience the em-algorithm seems to be more reliable, and for larger networks faster, than the numerical methods we used to calculate $\Phi_{CIS}$. Additionally, by letting the algorithm run multiple times we are able to gain a notion on how the local minima in $\mathcal{E}$ are related to each other as demonstrated in Figure \ref{DifflocMin}.

\section{Materials and Methods}

The distributions used in the Section \ref{sectResults} were generated by a python program and the measures $\Phi_{I}, \Phi_{CII}, \Phi_{SI}$ ans $\Phi_{G}$ are implemented in C++. The python package scipy.mimimize has been used to calculate $\Phi_{CIS}$. The code is available at Reference \cite{Code}.

\section*{Acknowledgement}
The authors acknowledge funding by Deutsche Forschungsgemeinschaft Priority Programme “The Active Self” (SPP 2134).


%

\appendix
\section{Graphical Models} \label{AppGraph}

Graphical models are a useful tool to visualize conditional independence structures. In this method a graph is used to describe the set of distributions that factor according to it. In our case, we are considering chain graphs.These are graphs, with vertex set $V$ and edge set $E \in V \times V$, consisting of directed and undirected edges such that we are able to partition the vertex set into subsets $V = V_{1} \cup \dots \cup V_{m}$, called chain components, with the properties that all edges between different subsets are directed, all edges between vertices of the same chain component are undirected and that there are no directed cycles between chain components. For a vertex set $\tau$, we will denote by $pa(\tau)$ the set of parents of element in $\tau$, which are vertices $\alpha$ with a directed arrow from  $\alpha$ to an element of $\tau$. Vertices connected by an undirected edge are called neighbours. A more detailed description can be found in Reference \cite{graphModels}.

\begin{Definition} \label{factor}
Let $T$ be the set of chain components.
A distribution factorizes with respect to a chain graph $G$ if the distribution can be written as follows
\begin{equation*} 
P(z) = \prod\limits_{\tau \in T} P(x_{\tau} \vert x_{pa(\tau)}),
\end{equation*}
where the structure of $P(x_{\tau} \vert x_{pa(\tau)})$ can be described in more detail.
Let $A(\tau), \tau \in T$ be the set of all subsets of $\tau \cup pa(\tau)$, that are complete in a graph $\tau_{\star}$, which is an undirected graph with the vertex set $\tau \cup pa(\tau)$ and the edges are the ones between elements in $\tau \cup pa(\tau)$ that exist in $G$ and additionally the ones between elements in $pa(\tau)$. An undirected graph is complete if every pair of distinct vertices is connected by an edge. Then there are non-negative functions $\phi_{a}$ such that
\begin{equation*}
P(x_{\tau} \vert x_{pa(\tau)}) = \prod\limits_{a \in A(\tau)} \phi_{a}(x).
\end{equation*}
\end{Definition}
If $\tau$ is a singleton then $\tau_{\star}$ is already complete. There are different kinds of independence statements a chain graph can encode, but we only need the global chain graph markov property. In order to define this property we need the concepts ancestral set and moral graph.

The boundary $bd(A)$ of a set $A \subseteq V$ is the set of vertices in $V \setminus A$ that are parents or neighbours to vertices in $A$. If $bd(\alpha) \subseteq A$ for all $\alpha \in A$ we call $A$ an ancestral set. For any $A \subseteq V$ there exists a smallest ancestral set containing $A$, because the intersection of ancestral sets is again an ancestral set. This smallest ancestral set of $A$ is denoted by $An(A)$.

Let $G$ be a chain graph. The moral graph of $G$ is an undirected graph denoted by $G^{m}$ that consists of the same vertex set as $G$ and in which two vertices $\alpha, \beta$ are connected if and only if either they were already connected by an edge in $G$ or if there are vertices $\gamma, \delta$ belonging to the same chain component such that $\alpha \rightarrow \gamma$ and $\beta \rightarrow \delta$.

\begin{Definition}[Global Chain Graph Markov Property]
Let $P$ be a distribution on $\mathcal{Z}$ and $G$ a chain graph. $P$ satisfies the global chain Markov property, with respect to $G$, if for any triple $(Z_{A}, Z_{B}, Z_{S})$ of disjoint subsets of $Z$ such that $Z_{S}$ separates $Z_{A}$ from $Z_{B}$ in $(G_{An(Z_{A} \cup Z_{B} \cup Z_{S})})^{m}$, the moral graph of the smallest ancestral set containing $Z_{A} \cup Z_{B} \cup Z_{S}$,
\begin{equation*}
Z_{A} \ci Z_{B} \mid Z_{S}
\end{equation*}
holds.
\end{Definition}

Since we are only considering positive discrete distributions, we have the following result.
\begin{Lemma} \label{equiv} The global chain Markov property and the factorization property are equivalent for positive discrete distributions.
\end{Lemma}
\begin{proof}[Proof of Lemma \ref{equiv}]
Theorem 4.1 from Reference \cite{Frydenberg} combined with the Hammersley–Clifford theorem, for example,~Theorem 2.9 in Reference \cite{nihat}, proves this statement.
\end{proof} 

In order to understand the conditional independence structure of a chain graph after marginalization, we need the following alogrithm from Reference \cite{Marginal}. This algorithm converts a chain graph with latent variables into a chain mixed graph with the conditional independence structure of the marginalized chain graph. A chain mixed graph has in addition to directed and undirected edges also bidirected edges, called arcs. The condition that there are no semi-directed cycles also applies to chain mixed graphs.
\begin{Definition} \label{Algorithm}
Let M be the set of vertices over which we want to marginalize. The following algorithm produces a chain mixed graph (CMG) with the conditional independence structure of the marginalized chain graph.

\begin{enumerate}
\item Generate an ij edge as in Table \ref{marg}, steps 8 and 9, between i and j on a collider trislide with an endpoint j and an endpoint in M if the edge of the same type does not already exist.
\item Generate an appropriate edge as in Table \ref{marg}, steps 1 to 7, between the endpoints of every tripath with inner node in M if the edge of the same type does not already exist. Apply this step until no other edge can be generated.
\item  Remove all nodes in M.
\end{enumerate}
\centering
\begin{tabular}{| c | c | c | l |}  
\hline
1 & i $\leftarrow$ m $\leftarrow$ j & generates & i $\leftarrow$ j \\
2 & i $\leftarrow$ m -- j & generates & i $\leftarrow$ j \\
3 & i $\leftrightarrow$ m ---j & generates & i $\leftrightarrow$ j \\
4 & i $\leftarrow$ m $\rightarrow$ j & generates & i $\leftrightarrow$ j \\
5 & i $\leftarrow$ m $\leftrightarrow$ j & generates & i $\leftrightarrow$ j \\
6 & i -- m $\leftarrow$ j & generates & i $\leftarrow$ j \\
7 & i -- m -- j & generates & i--j \\
\hline 
8 & m $\rightarrow$ i -- $\dots$  -- $\circ \leftarrow$ j & generates & i $\leftarrow$ j \\
9 & m $\rightarrow i -- \dots -- \circ \leftrightarrow$ j & generates & i $\leftrightarrow$ j \\
\hline  
\end{tabular} 
\captionof{table}{Types of edge induced by tripaths with inner node m $\in$ M and trislides with
endpoint m $\in$ M.} \label{marg} 
\end{Definition}
Conditional independence in CMGs is defined using the concept of c-separation, see for example Reference \cite{Marginal} in Section 4. For this definition we need the concepts of a walk and of a collider section. A walk is a list of vertices $\alpha_{0}, \dots, \alpha_{k}, \, k \in \mathbb{N}$, such there is an edge or arrow from $\alpha_{i}$ to $\alpha_{i+1}, \, i \in \{0, \dots , k-1\}$. A set of vertices connected by undirected edges is called a section. If there exists a walk including a section such that an arrow points at the first and last vertices of the section
\begin{equation*}
\rightarrow \bullet - \cdots  - \bullet \leftarrow
\end{equation*} 
then this is called a collider section.
\begin{Definition}[c-separation] \label{c-sep}
Let $A,B$ and $C$ be disjoint sets of vertices of a graph. A walk $\pi$ is called a c-connecting walk given $C$, if every collider section of $\pi$ has a node in $C$ and all non-collider sections are disjoint. The nodes $A$ and $B$ are called c-separated given $C$ if there are no c-connecting walks between them given $C$ and we write $A \ci_{c} B \vert C$.
\end{Definition}

\section{Proofs} \label{AppProof}
\begin{proof}[Proof of the Relationship (\ref{CIS})]
For $n=2$ this is immediate. Let now $n \geq 3$ and $i,j,k \in \{1, \dots ,n\}, \, i \neq j \neq k \neq i$. Applying (\ref{MarkovPop}) two times leads to
\begin{equation*}
\begin{split}
Q(y_{j},x) &= \dfrac{Q(y_{j}, x_{I \setminus \{i\}})Q(x)}{Q(x_{I \setminus \{i\}})} \\
Q(y_{j},x) &= \dfrac{Q(y_{j}, x_{I \setminus \{k\}})Q(x)}{Q(x_{I \setminus \{k\}})} \\
Q(y_{j}, x_{I \setminus \{i\}})Q(x_{I \setminus \{k\}}) &= Q(y_{j}, x_{I \setminus \{k\}})Q(x_{I \setminus \{i\}})
\end{split}
\end{equation*}
for all $(x,y_{j}) \in \mathcal{X} \times \mathcal{Y}_{j}$. Marginalizing over the elements of 
$\mathcal{X}_{k}$ yields
\begin{equation*}
\begin{split}
Q(y_{j}, x_{I \setminus \{i,k\}})Q(x_{I \setminus \{k\}}) &= Q(y_{j}, x_{I \setminus \{k\}})Q(x_{I \setminus \{i,k\}}) \\
Q(y_{j} \vert x_{I \setminus \{i,k\}}) &= Q(y_{j} \vert x_{I \setminus \{k\}}).
\end{split}
\end{equation*}
Using inductively the remaining relations results in (\ref{CIS}).
\end{proof}
{
\begin{proof}[Proof of Proposition \ref{IfandOnlyIf}]
If $\Phi_{CII}(\tilde{P}) = 0$ holds, then 
\begin{equation*}
 \inf\limits_{Q \in \mathcal{M}_{CII}} D_{\mathcal{Z}}(\tilde{P} \parallel Q) = 0.
\end{equation*} 
Since $\mathcal{M}_{CII}$ is compact the infimum is an element of  $\mathcal{M}_{CII}$, so there exists $Q \in \mathcal{M}_{CII}$ such that $D_{\mathcal{Z}}(P \parallel Q) = 0$. Therefore $P \in \mathcal{M}_{CII}$ and the existence of a sequence $Q^{m}$ follows from the definition of $\mathcal{M}_{CII}$.

Assume that there exists a sequence $Q^{m}$ that satisfies 1.~and 2. Then every element $Q^{m} \in \mathcal{M}_{CII}^{m}$ per definition and the limit
\begin{equation*}
\tilde{P} \in \overline{\bigcup\limits_{m \in \mathbb{N}} \mathcal{M}_{CII}^{m}} = \mathcal{M}_{CII}.
\end{equation*}
Hence 
\begin{equation*}
\Phi_{CII}(\tilde{P}) = \inf\limits_{Q \in \mathcal{M}_{CII}} D_{\mathcal{Z}}(\tilde{P} \parallel Q) = D_{\mathcal{Z}}(\tilde{P}, \tilde{P}) = 0.
\end{equation*} 
\end{proof}
}
\begin{proof}[Proof of Proposition \ref{PropGroundTruth}]

Let $P \in \mathcal{E}^{f}$ and  $Q \in \mathcal{E}$, then the KL-divergence between the two elements is
\begin{equation*}
\begin{split}
D_{\mathcal{Z} \times \mathcal{W}^{m}} (P \parallel Q) &=  \sum\limits_{z,w} P(z,w)\, log \, \dfrac{P(x)\prod\limits_{i}P(y_{i} \vert x, w) P(w)}{ Q(x) \prod\limits_{i} Q(y_{i} \vert x_{i}, w) Q(w)} \\
&= \sum\limits_{x} P(x) \, log \, \dfrac{P(x)}{Q(x)} + \sum\limits_{z,w} P(z,w) \, log \, \dfrac{\prod\limits_{i} P(y_{i} \vert x,w)}{ \prod\limits_{i} Q(y_{i} \vert x_{i}, w)} + \sum\limits_{w} P(w) \, log \, \dfrac{P(w)}{Q(w)} \\
&\geq \sum\limits_{x} P(x) \, log \, \dfrac{P(x)}{P(x)} + \sum\limits_{z,w} P(z,w) \, log \, \dfrac{\prod\limits_{i} P(y_{i} \vert x,w)}{\prod\limits_{i} P(y_{i} \vert x_{i}, w)} +  \sum\limits_{w} P(w) \, log \, \dfrac{P(w)}{P(w)}  \\
&= \sum\limits_{z,w} P(z,w) \, log \, \dfrac{\prod\limits_{i} P(y_{i} \vert x,w)}{\prod\limits_{i} P(y_{i} \vert x_{i}, w)}. 
\end{split}
\end{equation*} 
The inequality holds, because in the first and third addend, we are able to apply that the cross entropy is greater or equal to the entropy and in the second addend we use the log-sum inequality in the following way
\begin{equation*}
\begin{split}
\sum\limits_{z,w} P(z,w) \, log \, & \dfrac{\prod\limits_{i} P(y_{i} \vert x, w)}{ \prod\limits_{i} Q(y_{i} \vert x_{i}, w)} - \sum\limits_{z,w} P(z,w) \, log \, \dfrac{\prod\limits_{i} P(y_{i} \vert x, w)}{ \prod\limits_{i} P(y_{i} \vert x_{i}, w)} \\ &= \sum\limits_{x,w} P(x)P(w) \sum\limits_{y} \prod\limits_{i} P(y_{i}, \vert x,w) \, log \, \dfrac{\prod\limits_{i} P(y_{i} \vert x_{i}, w)}{ \prod\limits_{i} Q(y_{i} \vert x_{i}, w)} \\
&\geq \sum\limits_{x,w} P(x)P(w) \left( \sum\limits_{y} \prod\limits_{i} P(y_{i}, \vert x,w) \right) \, log \, \dfrac{ \sum\limits_{y} \prod\limits_{i} P(y_{i} \vert x_{i}, w)}{\sum\limits_{y} \prod\limits_{i} Q(y_{i} \vert x_{i}, w)} \\
&=0.
\end{split}
\end{equation*}
Therefore the new integrated information measure results in
\begin{equation*}
\inf\limits_{Q \in \mathcal{E}} D_{\mathcal{Z} \times \mathcal{W}^{m}} (P \parallel Q) = \sum\limits_{z,w} P(z,w) \, log \, \dfrac{\prod\limits_{i} P(y_{i} \vert x,w)}{\prod\limits_{i} P(y_{i} \vert x_{i}, w)}.
\end{equation*}
This can be rewritten to
\begin{equation*}
\begin{split}
\sum\limits_{z,w} P(z,w) \, log \, \dfrac{\prod\limits_{i} P(y_{i} \vert x,w)}{\prod\limits_{i} P(y_{i} \vert x_{i}, w)} &= \sum\limits_{z,w} P(z,w) \, log \, \dfrac{\prod\limits_{i} P(y_{i}, x,w) P(x_{i},w)}{\prod\limits_{i} P(y_{i}, x_{i}, w)P(x,w)} \\
 &= \sum\limits_{z,w} P(z,w) \, log \, \dfrac{\prod\limits_{i} P(y_{i}, x_{I \setminus \{i\}} \vert x_{i},w) P(x_{i}, w)}{\prod\limits_{i} P(y_{i} \vert x_{i}, w)P(x, w)} \\
 &= \sum\limits_{z,w} P(z,w) \, log \, \dfrac{\prod\limits_{i} P(y_{i},  x_{I \setminus \{i\}} \vert x_{i}, w) }{\prod\limits_{i} P(y_{i}\vert x_{i}, w) P(x_{I \setminus \{i\}} \vert x_{i},w)} \\
&= \sum\limits_{i} I(Y_{i} ; X_{I \setminus \{i\}} \vert X_{i}, W).
\end{split}
\end{equation*}
\end{proof}
\begin{proof}[Proof of Proposition \ref{relGroundTruth}]
By using the log-sum inequality we get
\begin{equation*}
\begin{split}
\Phi_{CII}^{m} &= \inf\limits_{Q \in \mathcal{M}_{CII}^{m}} \sum\limits_{z} P(z) log \dfrac{\sum\limits_{w} P(x) \prod\limits_{i} P(y_{i} \vert x,w) P(w)}{\sum\limits_{w} Q(x) \prod\limits_{i} Q(y_{i} \vert x_{i}, w) Q(w)} \\
&\leq  \inf\limits_{Q \in \mathcal{M}_{CII}^{m}} \sum\limits_{w} \sum\limits_{z} P(z,w) log \, \dfrac{P(x) \prod\limits_{i} P(y_{i} \vert x,w) P(w)}{ Q(x) \prod\limits_{i} Q(y_{i} \vert x_{i}, w) Q(w)} \\
&=\inf\limits_{Q \in \mathcal{E}} D_{\mathcal{Z} \times \mathcal{W}^{m}} (P \parallel Q).
\end{split}
\end{equation*}
The fact that every element of $Q \in \mathcal{E}$ corresponds via marginalization to an element in $\mathcal{M}_{CII}^{m}$ and every element in $\mathcal{M}_{CII}^{m}$ has at least one corresponding element in $Q \in \mathcal{E}$, leads to the equality in the last row. Since taking the infimum over a larger space can only decrease the value further, the relation
\begin{equation*}
\Phi_{CII} \leq \Phi_{T} 
\end{equation*}
holds.
\end{proof}
\begin{proof}[Proof of Proposition \ref{ColmProj}]
\begin{equation*}
\begin{split}
D_{\mathcal{Z} \times \mathcal{W}^{m}}(P \parallel Q) =& \sum\limits_{(z,w) \in \mathcal{Z} \times \mathcal{W}^{m}} P(z,w) \, log \, \dfrac{P(z,w)}{Q(x) \prod\limits_{i=1}^{n}Q(y_{i} \vert x_{i}, w)Q(w)} \\
=& \sum\limits_{(z,w) \in \mathcal{Z} \times \mathcal{W}^{m}} P(z,w) \, log \, P(z,w) \\
&+ \sum\limits_{(z,w) \in \mathcal{Z} \times \mathcal{W}^{m}} P(z,w) \, log \, \dfrac{1}{Q(x)} \\
&+ \sum\limits_{(z,w) \in \mathcal{Z} \times \mathcal{W}^{m}} \sum\limits_{i=1}^{n} P(z,w) \, log \, \dfrac{1}{Q(y_{i} \vert x_{i}, w)} \\
&+ \sum\limits_{(z,w) \in \mathcal{Z} \times \mathcal{W}^{m}} P(z,w) \, log \, \dfrac{1}{Q(w)}
\end{split}
\end{equation*}
The first addend is a constant for $P$ and the others are cross-entropies which are greater or equal to entropy
\begin{equation*}
\begin{split}
D_{\mathcal{Z} \times \mathcal{W}^{m}}(P \parallel Q) \geq& \sum\limits_{(z,w) \in \mathcal{Z} \times \mathcal{W}^{m}} P(z,w) \, log \, P(z,w) \\
&+ \sum\limits_{(z,w) \in \mathcal{Z} \times \mathcal{W}^{m}} P(z,w) \, log \, \dfrac{1}{P(x)} \\
&+ \sum\limits_{(z,w) \in \mathcal{Z} \times \mathcal{W}^{m}} \sum\limits_{i=1}^{n} P(z,w) \, log \, \dfrac{1}{P(y_{i} \vert x_{i}, w)} \\
&+ \sum\limits_{(z,w) \in \mathcal{Z} \times \mathcal{W}^{m}} P(z,w) \, log \, \dfrac{1}{P(w)} \\
=& \sum\limits_{(z,w) \in \mathcal{Z} \times \mathcal{W}^{m}} P(z,w) \, log \, \dfrac{P(z,w)}{P(x) \prod\limits_{i=1}^{n}P(y_{i} \vert x_{i}, w)P(w)}.
\end{split}
\end{equation*}
Therefore this projection is unique.
\end{proof}

\begin{proof}[Proof of Theorem \ref{proposition}]
We need a way to understand the connections in a graph after marginalization. In Reference \cite{Marginal} Sadeghi presents an algorithm that converts a chain graph to a chain mixed graph that represents the markov properties of the original graph after marginalizing, see Definition \ref{Algorithm}.

Although the actual set of distributions after marginalizing might be more complicated, it is a subset of the distributions factorizing according to the new graph, if the new graph is still a chain graph. This is due to the equivalence of the global chain Markov property and the factorization property in Lemma \ref{equiv}.

At first we will consider the case of two nodes per time step, $n = 2$. 
We will take a close look at the possible ways a hidden structure could be connected to the left graph in Figure \ref{startinggraph}. At first we will look at the possible connections between two nodes, depicted on the right in Figure \ref{startinggraph}. The boxes stand for any kind of subgraph of hidden nodes such that the whole graph is still a chain graph and the two headed dotted arrows stand for a line, or an arrow in any direction.
Consider two nodes $A$ and $B$, then the connections including a box between the nodes can take one of the five following forms
\begin{enumerate}
\item[(1)] they form an undirected path between $A$ and $B$,
\item[(2)] they can form a directed path from $A$ to $B$,
\item[(3)] they can form a directed path form $B$ to $A$,
\item[(4)] there exists a collider,
\item[(5)] $A$ and $B$ have a common exterior influence.
\end{enumerate}
A collider is a node or a set of nodes connected by undirected edges that have an arrow pointing at the set at both ends
\begin{equation*}
\rightarrow \bullet \cdots \bullet \leftarrow.
\end{equation*}
\begin{figure}[H]
\centering
\scalebox{1}{
\begin{tikzpicture}
\draw[line width=0.5mm] (-3,0) node {\footnotesize $X_{1}$};
\draw[line width=0.5mm] (-3,-1) node {\footnotesize $X_{2}$};
\draw[line width=0.5mm] (-0.9, 0) node {\footnotesize $Y_{1}$};
\draw[line width=0.5mm] (-0.9, -1) node {\footnotesize $Y_{2}$};
\draw[line width = 0.45mm] (-3,-0.3)--(-3,-0.7);
\draw[->,line width=0.45mm] (-2.7,0)--(-1.25,0);
\draw[->,,line width=0.45mm] (-2.7,-1)--(-1.25,-1);
\draw[line width=0.3mm] (-3,0) circle (8pt);
\draw[line width=0.3mm] (-3,-1) circle (8pt);
\draw[line width=0.3mm] (-0.9,0) circle (8pt);
\draw[line width=0.3mm] (-0.9,-1) circle (8pt);

\draw[line width=0.5mm] (1.5,0) node {\footnotesize $X_{1}$};
\draw[line width=0.5mm] (1.5,-1) node {\footnotesize $X_{2}$};
\draw[line width=0.5mm] (3.6, 0) node {\footnotesize $Y_{1}$};
\draw[line width=0.5mm] (3.6, -1) node {\footnotesize $Y_{2}$};
\filldraw [pattern = crosshatch dots] (2.3,0.6) rectangle (2.8,0.4);
\filldraw [pattern = crosshatch dots] (2.3,-1.6) rectangle (2.8,-1.4);

\filldraw [fill = gray] (2.1,-0.375) rectangle (2.3,-0.25);
\filldraw [fill = gray] (2.1,-0.625) rectangle (2.3,-0.75);

\filldraw [pattern = grid] (0.6,-0.35) rectangle (0.8,-0.65);
\filldraw [pattern = horizontal lines] (4.3,-0.35) rectangle (4.5,-0.65);
\draw[line width = 0.45mm] (1.5,-0.3)--(1.5,-0.7);
\draw[<->,dotted,line width=0.25mm] (1.2,-0.9)--(0.85,-0.65);
\draw[<->,dotted,line width=0.25mm] (1.2,-0.1)--(0.85,-0.35);
\draw[<->,dotted,line width=0.25mm] (3.3,0.15)--(2.85,0.4);
\draw[<->,dotted,line width=0.25mm] (3.3,-1.15)--(2.85,-1.4);
\draw[<->,dotted,line width=0.25mm] (1.8,0.15)--(2.25,0.4);
\draw[<->,dotted,line width=0.25mm] (1.8,-1.15)--(2.25,-1.4);
\draw[<->,dotted,line width=0.25mm] (3.9,-0.9)--(4.25,-0.65);
\draw[<->,dotted,line width=0.25mm] (3.9,-0.1)--(4.25,-0.35);
\draw[<->,dotted,line width=0.25mm] (3.3,-0.9)--(2.35,-0.375);
\draw[<->,dotted,line width=0.25mm] (2.05,-0.25)--(1.8,-0.1);
\draw[fill = white, white] (2.55,-0.5) circle (2pt);
\draw[<->,dotted,line width=0.25mm] (3.3,-0.1)--(2.35,-0.625);
\draw[<->,dotted,line width=0.25mm] (2.05,-0.75)--(1.8,-0.9);
\draw[->,,line width=0.45mm] (1.8,0)--(3.25,0);
\draw[->,,line width=0.45mm] (1.8,-1)--(3.25,-1);
\draw[line width=0.3mm] (1.5,0) circle (8pt);
\draw[line width=0.3mm] (1.5,-1) circle (8pt);
\draw[line width=0.3mm] (3.6,0) circle (8pt);
\draw[line width=0.3mm] (3.6,-1) circle (8pt);
\end{tikzpicture}}
\caption{Starting graph and possible two way interactions. } \label{startinggraph}
\end{figure}
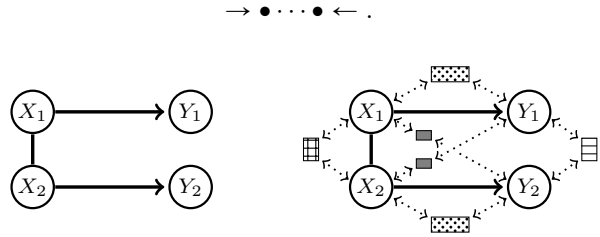
We will start with the gridded hidden structure connected to $X_{1}$ and $X_{2}$.
Since there already is an undirected edge between the $X_{i}$s an undirected path would make no difference in the marginalized model. The cases (2) and (3) would form a directed cycle which violates the requirements of a chain mixed graph. A collider would also make no difference, since it disappears in the marginalized model. A common exterior influence leads to  
\begin{equation*}
\begin{split}
P(\hat{w})P(x \vert \hat{w})P(y_{1} \vert x_{1}) P(y_{2} \vert x_{2}) &= P(x, \hat{w})P(y_{1} \vert x_{1}) P(y_{2} \vert x_{2}) \\
\sum\limits_{\hat{w}} P(x, \hat{w})P(y_{1} \vert x_{1}) P(y_{2} \vert x_{2}) &= P(x)P(y_{1} \vert x_{1}) P(y_{2} \vert x_{2}).
\end{split}
\end{equation*}

Now let us discuss these possibilities in the case of a gray hidden structure between $X_{i}$ and $Y_{j}$, $i,j \in \{1,2\},\, i \neq j$. An undirected edge or a directed edge (3) would create a directed cycle.  A directed path (2) from  $X_{i}$ to $Y_{j}$ would lead to a chain graph in which  $X_{i}$ and $Y_{j}$ are not conditionally independent given $X_{j}$. If there exists a collider (4) in the hidden structure, then nothing else in the graph depends on this part of the structure and it reduces to a factor one when we marginalize over the hidden variables. Therefore the path between $X_{i}$ and $Y_{j}$ gets interrupted leaving a potential external influence or effect. Those do not have an additional impact on the marginalized model. A common exterior influence (5) leads to a chain mixed graph which does not satisfy the necessary conditional independence structure, because using the Algorithm \ref{Algorithm} leads to an arc between $X_{i}$ and $Y_{j}$, hence they are c-connected in the sense of Definition \ref{c-sep}. 

The next possibility is a dotted hidden structure between $X_{i}$ and $Y_{i}, \, i \in \{1,2\}$. An undirected path (1) and a directed path (3) would lead to a directed cycle. A directed path (2) would add no new structure to the model since there already is a directed edge between $X_{i}$ and $Y_{i}$. A collider (4) does not have an effect on the marginalized model. Adding a common exterior influence $W_{1}$ on $X_{1}, Y_{1}$ results in a new model which is not symmetric in $i \in \{1,2\}$ and does not include $\mathcal{M}_{I}$, therefore it does not fully contain $\mathcal{M}_{CII}$. By adding additional common exterior $W_{2}$ influences on $X_{2}, Y_{2}$ or $Y_{1}, Y_{2}$, in order to include $\mathcal{M}_{I}$ in the new model, violates the conditional independence statements since nodes in $W_{1}$ and $W_{2}$ are connected in the moralized graph. 

The last hidden structure between two nodes is the striped one between the $Y_{i}$s. An undirected path (1) or any directed path (2),(3) lead to a graph that does not satisfy the conditional independence statements. A collider (4) has no impact on the model and a common exterior influence leads to the definition of Causal Information Integration. 

Connecting $Y_{1}, Y_{2} $ and $X_{i}, i \in \{1,2\}$ leads either to a violation of the conditional independence statements or contains a collider in which case the marginalized model reduces to one of the cases above. 

All the possible ways a hidden structure could be connected to three nodes $X_{1}, X_{2}, Y_{1}$ by directed edges are shown in Figure \ref{eight}. Replacing any of these edges by an undirected edge would either make no difference or lead to a model that does not satisfy the conditional independence statements. In this case the black boxes represent sections. More complicated hidden structures reduce to this case, since these structures either contain a collider and correspond to one of the cases above or contain longer directed paths in the direction of the edges connecting the structure to the visible nodes, which does not change the marginalized model. 
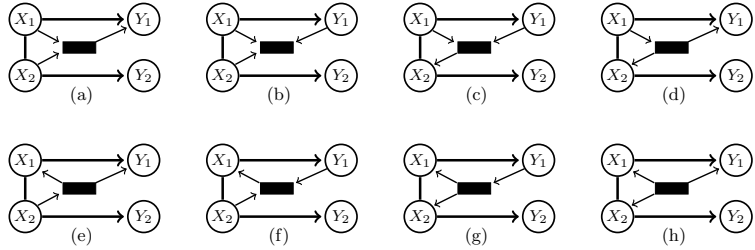
\begin{figure}[H]
\centering
\scalebox{0.75}{
\begin{tikzpicture}[']
\draw[] (-2,-1.35) node{ \small (a)};
\draw[line width=0.5mm] (-3,0) node {\footnotesize $X_{1}$};
\draw[line width=0.5mm] (-3,-1) node {\footnotesize $X_{2}$};
\draw[line width=0.5mm] (-0.9, 0) node {\footnotesize $Y_{1}$};
\draw[line width=0.5mm] (-0.9, -1) node {\footnotesize $Y_{2}$};
\filldraw [ draw=black] (-2.33,-0.6) rectangle (-1.75,-0.4);
\draw[line width = 0.45mm] (-3,-0.3)--(-3,-0.7);
\draw[->,,line width=0.25mm] (-2.775,-0.8)--(-2.4,-0.6);
\draw[->,,line width=0.25mm] (-2.775,-0.2)--(-2.4,-0.4);
\draw[->,,line width=0.25mm] (-1.75,-0.4)--(-1.2,-0.15);
\draw[->,,line width=0.45mm] (-2.7,0)--(-1.25,0);
\draw[->,,line width=0.45mm] (-2.7,-1)--(-1.25,-1);
\draw[line width=0.3mm] (-3,0) circle (8pt);
\draw[line width=0.3mm] (-3,-1) circle (8pt);
\draw[line width=0.3mm] (-0.9,0) circle (8pt);
\draw[line width=0.3mm] (-0.9,-1) circle (8pt);

\draw[] (1.5,-1.35) node{ \small (b)};
\draw[line width=0.5mm] (0.5,0) node {\footnotesize $X_{1}$};
\draw[line width=0.5mm] (0.5,-1) node {\footnotesize $X_{2}$};
\draw[line width=0.5mm] (2.6, 0) node {\footnotesize $Y_{1}$};
\draw[line width=0.5mm] (2.6, -1) node {\footnotesize $Y_{2}$};
\filldraw [ draw=black] (1.17,-0.6) rectangle (1.75,-0.4);
\draw[line width = 0.45mm] (0.5,-0.3)--(0.5,-0.7);
\draw[->,,line width=0.25mm] (0.725,-0.8)--(1.1,-0.6);
\draw[->,,line width=0.25mm] (0.725,-0.2)--(1.1,-0.4);
\draw[->,,line width=0.25mm] (2.35,-0.15)--(1.8,-0.4);
\draw[->,,line width=0.45mm] (0.8,0)--(2.25,0);
\draw[->,,line width=0.45mm] (0.8,-1)--(2.25,-1);
\draw[line width=0.3mm] (0.5,0) circle (8pt);
\draw[line width=0.3mm] (0.5,-1) circle (8pt);
\draw[line width=0.3mm] (2.6,0) circle (8pt);
\draw[line width=0.3mm] (2.6,-1) circle (8pt);

\draw[] (5,-1.35) node{ \small (c)};
\draw[line width=0.5mm] (4,0) node {\footnotesize $X_{1}$};
\draw[line width=0.5mm] (4,-1) node {\footnotesize $X_{2}$};
\draw[line width=0.5mm] (6.1, 0) node {\footnotesize $Y_{1}$};
\draw[line width=0.5mm] (6.1, -1) node {\footnotesize $Y_{2}$};
\filldraw [ draw=black] (4.67,-0.6) rectangle (5.25,-0.4);
\draw[line width = 0.45mm] (4,-0.3)--(4,-0.7);
\draw[<-,,line width=0.25mm] (4.25,-0.8)--(4.65,-0.6);
\draw[->,,line width=0.25mm] (4.225,-0.2)--(4.6,-0.4);
\draw[->,,line width=0.25mm] (5.85,-0.15)--(5.3,-0.4);
\draw[->,,line width=0.45mm] (4.3,0)--(5.75,0);
\draw[->,,line width=0.45mm] (4.3,-1)--(5.75,-1);
\draw[line width=0.3mm] (4,0) circle (8pt);
\draw[line width=0.3mm] (4,-1) circle (8pt);
\draw[line width=0.3mm] (6.1,0) circle (8pt);
\draw[line width=0.3mm] (6.1,-1) circle (8pt);

\draw[] (8.5,-1.35) node{ \small (d)};
\draw[line width=0.5mm] (7.5,0) node {\footnotesize $X_{1}$};
\draw[line width=0.5mm] (7.5,-1) node {\footnotesize $X_{2}$};
\draw[line width=0.5mm] (9.6, 0) node {\footnotesize $Y_{1}$};
\draw[line width=0.5mm] (9.6, -1) node {\footnotesize $Y_{2}$};
\filldraw [ draw=black] (8.17,-0.6) rectangle (8.75,-0.4);
\draw[line width = 0.45mm] (7.5,-0.3)--(7.5,-0.7);
\draw[<-,,line width=0.25mm] (7.75,-0.8)--(8.15,-0.6);
\draw[->,,line width=0.25mm] (7.725,-0.2)--(8.1,-0.4);
\draw[<-,,line width=0.25mm] (9.3,-0.15)--(8.75,-0.4);
\draw[->,,line width=0.45mm] (7.8,0)--(9.25,0);
\draw[->,,line width=0.45mm] (7.8,-1)--(9.25,-1);
\draw[line width=0.3mm] (7.5,0) circle (8pt);
\draw[line width=0.3mm] (7.5,-1) circle (8pt);
\draw[line width=0.3mm] (9.6,0) circle (8pt);
\draw[line width=0.3mm] (9.6,-1) circle (8pt);

\draw[] (-2,-3.85) node{ \small (e)};
\draw[line width=0.5mm] (-3,-2.5) node {\footnotesize $X_{1}$};
\draw[line width=0.5mm] (-3,-3.5) node {\footnotesize $X_{2}$};
\draw[line width=0.5mm] (-0.9, -2.5) node {\footnotesize $Y_{1}$};
\draw[line width=0.5mm] (-0.9, -3.5) node {\footnotesize $Y_{2}$};
\filldraw [ draw=black] (-2.33,-3.1) rectangle (-1.75,-2.9);
\draw[line width = 0.45mm] (-3,-2.8)--(-3,-3.2);
\draw[->,,line width=0.25mm] (-2.775,-3.3)--(-2.4,-3.1);
\draw[<-,,line width=0.25mm] (-2.7,-2.7)--(-2.35,-2.9);
\draw[->,,line width=0.25mm] (-1.75,-2.9)--(-1.2,-2.65);
\draw[->,,line width=0.45mm] (-2.7,-2.5)--(-1.25,-2.5);
\draw[->,,line width=0.45mm] (-2.7,-3.5)--(-1.25,-3.5);
\draw[line width=0.3mm] (-3,-2.5) circle (8pt);
\draw[line width=0.3mm] (-3,-3.5) circle (8pt);
\draw[line width=0.3mm] (-0.9,-2.5) circle (8pt);
\draw[line width=0.3mm] (-0.9,-3.5) circle (8pt);

\draw[] (1.5,-3.85) node{ \small (f)};
\draw[line width=0.5mm] (0.5,-2.5) node {\footnotesize $X_{1}$};
\draw[line width=0.5mm] (0.5,-3.5) node {\footnotesize $X_{2}$};
\draw[line width=0.5mm] (2.6, -2.5) node {\footnotesize $Y_{1}$};
\draw[line width=0.5mm] (2.6, -3.5) node {\footnotesize $Y_{2}$};
\filldraw [ draw=black] (1.17,-3.1) rectangle (1.75,-2.9);
\draw[line width = 0.45mm] (0.5,-2.8)--(0.5,-3.2);
\draw[->,,line width=0.25mm] (0.725,-3.3)--(1.1,-3.1);
\draw[<-,,line width=0.25mm] (0.775,-2.7)--(1.15,-2.9);
\draw[->,,line width=0.25mm] (2.35,-2.65)--(1.8,-2.9);
\draw[->,,line width=0.45mm] (0.8,-2.5)--(2.25,-2.5);
\draw[->,,line width=0.45mm] (0.8,-3.5)--(2.25,-3.5);
\draw[line width=0.3mm] (0.5,-2.5) circle (8pt);
\draw[line width=0.3mm] (0.5,-3.5) circle (8pt);
\draw[line width=0.3mm] (2.6,-2.5) circle (8pt);
\draw[line width=0.3mm] (2.6,-3.5) circle (8pt);

\draw[] (5,-3.85) node{ \small (g)};
\draw[line width=0.5mm] (4,-2.5) node {\footnotesize $X_{1}$};
\draw[line width=0.5mm] (4,-3.5) node {\footnotesize $X_{2}$};
\draw[line width=0.5mm] (6.1, -2.5) node {\footnotesize $Y_{1}$};
\draw[line width=0.5mm] (6.1, -3.5) node {\footnotesize $Y_{2}$};
\filldraw [ draw=black] (4.67,-3.1) rectangle (5.25,-2.9);
\draw[line width = 0.45mm] (4,-2.8)--(4,-3.2);
\draw[<-,,line width=0.25mm] (4.25,-3.3)--(4.65,-3.1);
\draw[<-,,line width=0.25mm] (4.275,-2.7)--(4.65,-2.9);
\draw[->,,line width=0.25mm] (5.85,-2.65)--(5.3,-2.9);
\draw[->,,line width=0.45mm] (4.3,-2.5)--(5.75,-2.5);
\draw[->,,line width=0.45mm] (4.3,-3.5)--(5.75,-3.5);
\draw[line width=0.3mm] (4,-2.5) circle (8pt);
\draw[line width=0.3mm] (4,-3.5) circle (8pt);
\draw[line width=0.3mm] (6.1,-2.5) circle (8pt);
\draw[line width=0.3mm] (6.1,-3.5) circle (8pt);

\draw[] (8.5,-3.85) node{ \small (h)};
\draw[line width=0.5mm] (7.5,-2.5) node {\footnotesize $X_{1}$};
\draw[line width=0.5mm] (7.5,-3.5) node {\footnotesize $X_{2}$};
\draw[line width=0.5mm] (9.6, -2.5) node {\footnotesize $Y_{1}$};
\draw[line width=0.5mm] (9.6, -3.5) node {\footnotesize $Y_{2}$};
\filldraw [ draw=black] (8.17,-3.1) rectangle (8.75,-2.9);
\draw[line width = 0.45mm] (7.5,-2.8)--(7.5,-3.2);
\draw[<-,,line width=0.25mm] (7.75,-3.3)--(8.15,-3.1);
\draw[<-,,line width=0.25mm] (7.775,-2.7)--(8.15,-2.9);
\draw[<-,,line width=0.25mm] (9.3,-2.65)--(8.75,-2.9);
\draw[->,,line width=0.45mm] (7.8,-2.5)--(9.25,-2.5);
\draw[->,,line width=0.45mm] (7.8,-3.5)--(9.25,-3.5);
\draw[line width=0.3mm] (7.5,-2.5) circle (8pt);
\draw[line width=0.3mm] (7.5,-3.5) circle (8pt);
\draw[line width=0.3mm] (9.6,-2.5) circle (8pt);
\draw[line width=0.3mm] (9.6,-3.5) circle (8pt);
\end{tikzpicture}}
\caption{The eight possible hidden structures between three nodes. } \label{eight}
\end{figure}
The models in (c), (d), (e), (f) and (g) contain either a collider and reduce therefore to one of the cases discussed above or induce a directed cycle.
We see that (a) and (h) display structures that do not satisfy the conditional independence statements. The hidden structure in (b) has no impact on the model.

A hidden structure connected to all four nodes contains one of the structures above and therefore does not induce a new valid model. 

Let us now consider a model with $n>2$. Any hidden structure on this model either connects only up to four nodes and reduces therefore to one of the cases above, contains one of the connections discussed in Figure \ref{eight} or only connects nodes among one point in time. The only structures possible to add would be a common exterior influence on the $X_{i}$s, a common exterior influence on the $Y_{i}$s or a collider section on any nodes. All these structures do not change the marginalized model. Therefore it is not possible to create a chain graph with hidden nodes in order to get a model strictly larger than $\mathcal{M}_{CII}$.
\end{proof}

\printbibliography

\end{document}